\documentclass[sigconf]{acmart}

\usepackage{amsmath,amsthm,mathrsfs}
\usepackage[ruled,vlined]{algorithm2e}
\usepackage{dsfont}
\usepackage{algorithmic}
\usepackage{graphicx}
\usepackage{textcomp}
\usepackage{xcolor}
\usepackage{thmtools}

\usepackage{float}
\usepackage{booktabs}
\usepackage{multirow}
\usepackage{array}
\usepackage{threeparttable}
\usepackage{siunitx}

\begin{document}

\title{On the Reachability Problem for One-Dimensional Thin Grammar Vector Addition Systems}

\author{Chengfeng Xue}
\email{xcf123@sjtu.edu.cn}
\affiliation{
  \institution{Shanghai Jiao Tong University}
  \city{Shanghai}
  \state{}
  \country{China}
}

\author{Yuxi Fu}
\email{fu-yx@cs.sjtu.edu.cn}
\affiliation{
  \institution{Shanghai Jiao Tong University}
  \city{Shanghai}
  \state{}
  \country{China}
}

\begin{abstract}
    Vector addition systems with states (VASS) are a classic model in concurrency theory. Grammar vector addition systems (GVAS), equivalently, pushdown VASS, extend VASS by using a context-free grammar to control addition. In this paper, our main focus is on the reachability problem for one-dimensional thin GVAS (thin 1-GVAS), a structurally restricted yet expressive subclass. By adopting the index measure for complexity, and by generalizing the decomposition technique developed in the study of VASS reachability to grammar-generated derivation trees of GVAS, an effective integer programming system is established for a thin 1-GVAS. 
    In this way, a nondeterministic algorithm with $\mathbf{F}_{2k}$ complexity is obtained for the reachability of thin 1-GVAS with index $k$, yielding a tighter upper bound than the previous one.
\end{abstract}



\maketitle
\iffalse
%%%%%%%%%%%%%%%%%%%%%%%%%
\newtheorem{theorem}{Theorem}[section]
\newtheorem{lemma}[theorem]{Lemma}
\newtheorem{definition}[theorem]{Definition}
\newtheorem{proposition}[theorem]{Proposition}
\newtheorem{corollary}[theorem]{Corollary}

\theoremstyle{definition}
\newtheorem{remark}[theorem]{Remark}
\newtheorem{example}[theorem]{Example}
%%%%%%%%%%%%%%%%%%%%%%%%%
\fi

\newcommand{\mrm}[1]{\mathrm{#1}}

\newcommand{\norm}[1]{\left\Vert{#1}\right\Vert}
\newcommand{\norminf}[1]{\left\Vert{#1}\right\Vert_{\infty}}
\newcommand{\PSPACE}{{\textbf{PSPACE}}}
\newcommand{\EXPSPACE}{{\textbf{EXPSPACE}}}
\newcommand{\EXP}{{\mrm{exp}}}
\newcommand{\TWOEXP}{{\mrm{2\text{-}exp}}}

\newcommand{\NP}{{\textbf{NP}}}
\newcommand{\FF}{\mathbf{F}}
\newcommand{\FFF}{\mathscr{F}}

\newcommand{\bbn}{\mathbb{N}}
\newcommand{\bbnw}{\bbn_{\omega}}
\newcommand{\bbnww}{\bbn_{\pm \omega}}
\newcommand{\bbz}{\mathbb{Z}}
\newcommand{\bbq}{\mathbb{Q}}
\newcommand{\para}[1]{\vspace{1mm}\textbf{#1}}
\newcommand{\Chi}{\mathcal{X}}
\newcommand{\Rho}{\mathcal{P}}
\newcommand{\defeq}{\stackrel{\scriptscriptstyle{\mathrm{def}}}{=}}

\renewcommand{\Vec}[1]{\mathbf{#1}}
\renewcommand{\vec}[1]{\mathbf{#1}}
\newcommand{\tl}[1]{\widetilde{#1}}
\newcommand{\sol}{\mathrm{sol}}
\newcommand{\solmin}{\mathrm{sol}_\mathrm{min}}
\newcommand{\solh}{\mathrm{Sol}_{H}}
\newcommand{\solt}{\mathrm{sol}_\mathcal{T}}
\newcommand{\ccf}{\theta_1(\CC)}
\newcommand{\ccb}{\theta_2(\CC)}
\newcommand{\ccr}{\theta_3(\CC)}

\newcommand{\gdim}[1]{\mathrm{dim}{(#1)}}
\newcommand{\rank}[1]{\mathrm{rank}{(#1)}}

\newcommand{\reach}{\mrm{\mathbf{Reach}}}
\newcommand{\cover}{\mrm{\mathbf{Cover}}}

\newcommand{\size}[1]{\mathrm{size}{(#1)}}
\newcommand{\poly}[1]{\mathrm{poly}{(#1)}}
\newcommand{\sub}[1]{\mathrm{sub}{(#1)}}
\newcommand{\ppath}[1]{\mathrm{path}{(#1)}}
\newcommand{\leaf}[1]{\mathrm{Leaf}{(#1)}}
\newcommand{\seg}[1]{\mathrm{seg}{(#1)}}
\newcommand{\cyc}[1]{\mathrm{cyc}{(#1)}}
\newcommand{\subdiv}[1]{\mathrm{Div}{(#1)}}
\newcommand{\subdivs}[1]{\mathrm{Div}'{(#1)}}

\newcommand{\indeg}[1]{\mrm{deg_{in}}{(#1)}}
\newcommand{\outdeg}[1]{\mrm{deg_{out}}{(#1)}}

\newcommand{\CC}{\mathcal{C}}
\newcommand{\PG}{\tl{G}}
\newcommand{\GG}{\mathcal{G}}
\newcommand{\EE}{\mathcal{E}}

\newcommand{\TT}{\mathcal{T}}
\newcommand{\NTT}{\mathcal{T}'}

\newcommand{\KT}{\mathcal{KT}}
\newcommand{\NKT}{\mathcal{KT}'}

\newcommand{\constr}{\mrm{\mathbf{Constr}}}
\newcommand{\ortho}{\mrm{\mathbf{Ortho}}}
\newcommand{\clean}{\mrm{\mathbf{Clean}}}
\newcommand{\decompA}{\mrm{\mathbf{Decomp_A}}}
\newcommand{\decompC}{\mrm{\mathbf{Decomp_C}}}
\newcommand{\refA}{\mrm{\mathbf{Ref_A}}}
\newcommand{\refC}{\mrm{\mathbf{Ref_C}}}

\newcommand{\Bounded}{\mrm{B}}
\newcommand{\RhoB}{\Rho_{\mrm{\Bounded}}}
\newcommand{\Pump}{\mrm{Pump}(\CC)}

\newcommand{\RhoSCC}{\Rho_{\mrm{SCC}}}
\newcommand{\RhoL}{\Rho_{\mrm{L}}}
\newcommand{\RhoR}{\Rho_{\mrm{R}}}
\newcommand{\rhoD}{\rho_{\mrm{D}}}
\newcommand{\ChiSCC}{\Chi_{\mrm{SCC}}}
\newcommand{\ChiL}{\Chi_{\mrm{L}}}
\newcommand{\ChiR}{\Chi_{\mrm{R}}}
\newcommand{\SigmaL}{\Sigma_{\mrm{L}}}
\newcommand{\SigmaR}{\Sigma_{\mrm{R}}}
\newcommand{\lP}{\vec l_{\mrm{src}}}
\newcommand{\lQ}{\vec l_{\mrm{tgt}}}
\newcommand{\rP}{\vec r_{\mrm{src}}}
\newcommand{\rQ}{\vec r_{\mrm{tgt}}}
\newcommand{\lPC}{\lP^\CC}
\newcommand{\lQC}{\lQ^\CC}
\newcommand{\rPC}{\rP^\CC}
\newcommand{\rQC}{\rQ^\CC}

\newcommand{\capture}{\rhd}

\section{Introduction}
\label{section:intro}
Vector addition systems with states (VASS) are a fundamental model in concurrency theory. The most widely studied problem is the reachability problem, addressing the issue of whether an initial configuration can reach a target configuration via a nonnegative run. 

The reachability problem for VASS was recently shown to be Ackermann-complete~\cite{DBLP:conf/icalp/FuYZ24_F_d_KLMST,DBLP:conf/stoc/CzerwinskiLLLM19_not_elementary,DBLP:conf/focs/CzerwinskiO21_ackermannian_complete,DBLP:conf/fsttcs/CzerwinskiJ0LO23_Lowerbound, DBLP:conf/lics/LerouxS19_KLMST}. For VASS parameterized by dimension $d$ ($d$-VASS), the reachability problem is in $\FF_d$ \cite{DBLP:conf/icalp/FuYZ24_F_d_KLMST}, where $\FF_d$ denotes the fast-growing hierarchy defined in \cite{DBLP:journals/toct/Schmitz16_Hierarchy}. The reachability problem is solved using the \emph{KLM algorithm}, which was first introduced in the 1980s~\cite{Kosaraju,LAMBERT199279,mayr1981} and continuously refined over decades~\cite{DBLP:conf/icalp/FuYZ24_F_d_KLMST,DBLP:conf/lics/LerouxS19_KLMST,Leroux2015demystifying}. Regarding the lower bound, Czerwi\'nski et al.\ proved the $\FF_d$ hardness for $(2d+3)$-dimensional VASS through iterative refinements~\cite{DBLP:conf/focs/CzerwinskiO21_ackermannian_complete,DBLP:conf/fsttcs/CzerwinskiJ0LO23_Lowerbound}. Although \PSPACE{}-completeness \cite{DBLP:conf/concur/Zheng25_geo_2_d} has been established for 2-VASS, and 2-\EXPSPACE{} upper bound \cite{DBLP:conf/icalp/CzerwinskiJ0O25_3VASS} has been obtained for 3-VASS, the  gap between the upper and bounds remains huge for higher-dimensional VASS reachability problems.

Many variants of VASS have been investigated. 
Grammar vector addition systems (GVAS) are vector addition systems controlled by context-free grammars (CFGs). A $d$-dimensional GVAS  ($d$-GVAS) is a CFG whose terminals are vectors in $\bbz^d$. Pushdown VASS (PVASS) extends VASS with a pushdown stack. By the classical equivalence between pushdown automata and CFGs, GVAS and PVASS are equivalent. The coverability, boundedness and reachability problems of GVAS have been studied. Leroux, Sutre and Totzke established an \textbf{EXPTIME} upper bound for the boundedness problem of 1-GVAS \cite{DBLP:conf/rp/LerouxST15_Boundedness}, and an \EXPSPACE{} upper bound for the coverability of 1-GVAS \cite{DBLP:conf/icalp/LerouxST15_Coverability}. Bizi{\`{e}}re and Czerwi\'nski showed the decidability of the reachability of 1-GVAS \cite{DBLP:conf/stoc/BiziereC25_BC25}, and later, Guttenberg, Keskin and Meyer proved the decidability of $d$-GVAS reachability for arbitrary dimensions \cite{DBLP:journals/corr/abs-2504-05015_PVASS_decidable}. 

There is a clear dividing line between the hardness of different GVASes, known as \emph{thinness}. A GVAS is \emph{thin}, if no nonterminal can produce more than one copy of itself, i.e., there exists no derivation of the form $X\Rightarrow \alpha X\beta X\gamma$. Thinness is closely related to the \emph{index}, introduced by Atig and Ganty~\cite{DBLP:conf/fsttcs/AtigG11_AG11}. A GVAS is $k$-indexed, if in every derivation, there exists a derivation in which at most $k$ nonterminals appear simultaneously. A GVAS is thin if and only if it is finitely indexed. The significance of thinness lies in the fact that many studies reduce general GVAS to thin GVAS \cite{DBLP:conf/stoc/BiziereC25_BC25,DBLP:conf/icalp/LerouxST15_Coverability}. 
Atig and Ganty also proved the decidability of the reachability of thin GVAS~\cite{DBLP:conf/fsttcs/AtigG11_AG11}.

The best-known upper bound for thin GVAS is derived from \cite{DBLP:conf/lics/GuttenbergCL25_CmeVASS}. Guttenberg, Czerwi\'nski and Lasota generalized the edge relations of VASS from vector additions to reachability relations, and introduced a universal algorithm to approximate the reachability sets of such extended VASS. As a consequence, the reachability problem for thin $d$-GVAS indexed by $k$ is in $\FF_{4kd+2k-4d}$ via reduction to their extended model. For thin 1-GVAS, the upper bound is $\FF_{6k-4}$. Furthermore, the reduction from 1-GVAS to thin 1-GVAS proposed in \cite{DBLP:conf/stoc/BiziereC25_BC25} implies that the reachability problem for general 1-GVAS is Ackermannian. Whether there exists a better upper bound for thin 1-GVAS (and consequently for 1-GVAS), or even an elementary one, remains open.

Compared to VASS, the reachability problem for GVAS presents several significant challenges. Firstly, it is a folklore that the reachability problem for $d$-GVAS can be reduced to the coverability problem for $(d+1)$-GVAS. When investigating GVAS reachability, simple properties of coverability, such as the existence of small coverability certificates for $d\geq 2$, are not available. Secondly, the finite reachability set of $d$-GVAS is hyper-Ackermannian \cite{DBLP:conf/csl/LerouxPS14_Hyperackermannian}. An Ackermannian reachability set can already be constructed for thin 1-GVAS. Hence, any method relying on brute-force enumeration of reachability certificates (e.g., the KLM algorithm) cannot overcome the Ackermannian barrier, even for thin 1-GVAS.

\emph{Our contribution.} In the present paper, we study the reachability problem for thin 1-GVAS. Our main contribution is the following theorem.

\begin{restatable}{theorem}{TheoremTwok}
    The reachability problem for thin 1-GVAS indexed by $k$ is in $\FF_{2k}$.
    \label{theorem:2k}
\end{restatable}

Since $k\geq 2$ for all nontrivial GVAS, the $\FF_{2k}$ upper bound improves upon the previous $\FF_{6k-4}$ result. To achieve this, we generalize the KLM decomposition for VASS to thin GVAS. The core idea of the KLM decomposition is as follows. A nonnegative run ($\bbn$-run) of VASS can be divided into many segments, each of which forms a strongly-connected component. For each such component, an integer programming system--the characteristic system--is constructed to describe the $\bbz$-run, where non-negativity is not ensured.
However, as long as the characteristic system satisfies a certain condition, called perfectness, one can construct an $\bbn$-run from a $\bbz$-run. To ensure perfectness, decompositions are performed on components iteratively. The termination of decomposition is guaranteed using a suitable \emph{ranking function}.

We apply the above approach to thin GVAS by dividing the derivation tree into exponentially many ``strongly-connected'' segments. To capture the derivation behavior of these segments, we introduce the \emph{KLM tree}, which can be applied to thin GVAS of arbitrary dimension. Furthermore, we define a ranking function that accounts for both the dimension and the index. In the light of the difficulty of the coverability problem and the pumping techniques for GVAS, we further restrict our attention to the one-dimensional case in decomposition step. We employ graph-theoretic techniques together with coverability results for thin 1-GVAS, and finally bound the length of the decomposition sequence by~$\FF_{2k}$, which yields our conclusion.

\emph{Outline.}
Section~\ref{section:preliminaries} introduces the notation and preliminaries.
Section~\ref{section:GVAS} presents thin GVAS and their index.
Section~\ref{section:KLM_tree} introduces KLM trees as a structural representation of reachability.
Section~\ref{section:decomposing_KLM_tree} develops refinement operations leading to perfect KLM trees.
Section~\ref{section:perfect_KLM_tree} establishes small perfect KLM trees as reachability certificates, from which our main result follows.

\section{Preliminaries}
\label{section:preliminaries}

We use $\bbn,\bbn_+,\bbz,\bbq$ to denote the set of natural numbers, positive natural numbers, integers, and rational numbers respectively. 
Let $\bbnw \defeq \bbn \cup \{+\omega\}, \bbnww \defeq \bbn \cup \{+\omega, -\omega\}$, where $\omega$ stands for unbounded element.
We postulate that $\forall x\in \bbn, -\omega <x<+\omega$ and $-\omega \pm x = -\omega, +\omega \pm x = +\omega$. Let $x,y\in \bbn$. 
We use $[x,y]$ to denote the number set $\{x,x+1,...,y\}$ if $y\geq x$. Otherwise, $[x,y] = \varnothing$. 
We use bold lowercase letters to denote vectors. 
For any $\Vec{v}\in \bbq^d$, let $\Vec{v}[i]$ be the $i$-th coordinate of $\Vec{v}$. Given $\Vec{v},\vec{u} \in \bbq^d$, we have $\vec{v}\leq \vec{u}$ if $\vec{v}[i]\leq \vec{u}[i]$ for every $i\in [1,d]$. 
We denote by $\norm{\Vec{v}}_1, \norm{\Vec{v}}_\infty$ the $l_1$-norm and the $l_\infty$-norm of $\Vec{v}$ respectively.

We use uppercase letters for sets. Given sets $S,T$, by $S \setminus T$ we denote the operator $S$ minus $T$. Unless specified otherwise, $|S|$ denotes the cardinality of the set $S$. Let $S \subseteq \bbz^d$ be a set of vectors. The $\bbq$-generated set of $S$ is:
\[\bbq(S) \defeq \left\{\sum_{\vec v_i\in S} \lambda_i \cdot \vec v_i\mid \lambda_i\in \bbq \right\}\]
which is also called the vector space spanned by $S$. The definitions of $\bbn(S), \bbz(S)$ are similar. The geometric dimension $\gdim{S}$ is defined as the dimension of the vector space $\bbq(S)$. If $S = \{\vec 0\}$ or $S = \varnothing$, the geometric dimension is $0$. Let $S\subseteq \bbz$ be a set of integers. We define $\sup S$ as the supremum of $S$. If $S$ is unbounded above, we write $\sup S = +\omega$.

In this paper, let $\poly{x}$ be any feasible polynomial function of $x$ according to the context. We abbreviate $2^{\poly{x}}$ to $\EXP(x)$. The fast-growing functions are given as follow. $F_1 : \bbn \rightarrow \bbn, F_1(n) \defeq 2n$ and $F_d :\bbn \rightarrow \bbn, F_d(n) \defeq F_{d-1}^{n} (1)$ for $d>1$. It is clear that $F_2(n) = 2^n$ and $F_3(n) = \mathrm{TOWER}(n)$. The Ackermann function is defined through diagonalization. $F_{\omega} :\bbn \rightarrow \bbn, F_{\omega}(n) \defeq F_{n}(n)$. We introduce the extended Grzegorczyk hierarchy \cite{DBLP:journals/toct/Schmitz16_Hierarchy}. The $d$-th level is the class of functions computable by deterministic Turing machines in $O(F_{d}^c(n))$ time where $n$ is the input size and $c$ is a constant. Formally, we have:
\[\FFF_d \defeq \bigcup_{c<\omega} \textbf{FDTIME}(F_d^c(n)).\]
We define the fast-growing hierarchy of decision problems. For $d\geq 2$, we have: \[\FF_d \defeq \bigcup_{p\in \FFF_{<d}} \textbf{DTIME}(F_d(p(n))).\] Clearly, $\FFF_d$ is closed under finite many compositions of functions $F_{\leq d}$. Moreover, in the definition of $\FF_d$, the choice of \textbf{DTIME} or \textbf{NTIME} is irrelevant for $d\geq 3$. Function $\exp(n)$ is in $\FFF_2$ and the complexity class \textbf{ELEMENTARY} is in $\FF_3$.

The fast-growing hierarchy provides length bounds for \emph{bad} sequences. 
Given a sequence $\vec x_1,...,\vec x_l$ with $\vec x_i\in \bbn^n$. The sequence is bad if there is no such pair $i<j$ with $\vec x_i\leq \vec x_j$. The sequence is $(f,m)$-controlled for function $f$ and natural number $m$ if $\norm{\vec x_i}_\infty \leq f^i(m)$.
The following lemma is from~\cite{5970250_controlled_seq} and the particular version stated below is from~\cite{DBLP:conf/lics/GuttenbergCL25_CmeVASS}.

\begin{lemma}[Bad sequences]
\label{lemma:bad_seq}
Let $f\in\FFF_\gamma$ be monotone and $f(x)\geq x$. The length function of any $(f,m)$-controlled bad sequence about $m$ is in $\FFF_{\gamma+n-1}$.
\end{lemma}

Pottier's Lemma~\cite{Pottier} provides a size bound for the minimal solutions of integer programming. The following statement (see \cite{DBLP:conf/icalp/CzerwinskiJ0O25_3VASS}, Lemma 4.) covers only the equality case. An integer programming system with inequalities can be transformed into equalities by introducing slack variables.
\begin{lemma}[Pottier]\label{lemma:pottier}
Suppose $\mathbf{A}\cdot \vec x = \vec b$ is a Diophantine linear equation system  with $n$ variables, $m$ equations.
If the absolute values of the coefficients is bounded by $N$, then $\norminf{\solmin}=O(nN)^m$ for every minimal solution $\solmin$.
\end{lemma}

\section{Grammar Vector Addition Systems}
\label{section:GVAS}
A \emph{$d$-dimensional grammar vector addition system} ($d$-GVAS) is a context-free grammar where all terminals are integer vectors. It is defined as the tuple: \[\GG\defeq(\Chi,\Sigma,\Rho,S_0),\] where $\Chi$ is a finite set of nonterminals, $\Sigma \subseteq \bbz^d$ is a finite set of terminal vectors, $\Rho\subseteq \Chi \times (\Sigma \cup\Chi^2 )$ is a finite set of production rules in Chomsky normal form, and $S_0\in \Chi$ is the starting symbol. 
The start symbol $S_0$ may may appear on the right side of a production rule.

We use uppercase letters $A,B,X,Y,\dots$ to denote nonterminals in $\Chi$. The distinction between nonterminals and sets is clear from the context. Terminal vectors in $\Sigma$ are denoted as  $\Vec{u},\Vec{v},\dots$. production rules in $\Rho$ are written as $X\rightarrow AB$ or $X\rightarrow \Vec{u}$. A derivation $X\Rightarrow \alpha$ is obtained through application of production rules where $\alpha \in \{\Chi \cup \Sigma\}^*$ is a string of terminals and nonterminals. A derivation is \emph{complete} if $\alpha \in \Sigma^*$.

A GVAS is \emph{proper} if every symbol can be derived from $S_0$, and every nonterminal has a complete derivation. For the rest of the paper, we restrict our attention to the proper GVAS. We will investigate subsets of $\GG$. For $X\in \Chi$, let $\GG_X = (\Chi', \Sigma',\Rho',X)$ be the proper sub-GVAS induced by $X$. $\Rho'$ collect all the production rules that can be derived from $X$, and $\Chi', \Sigma'$ are obtained accordingly.

\emph{The size of GVAS.} Given $\GG = (\Chi,\Sigma,\Rho,S_0)$, let $|\GG|$ be the total number of symbols and rules. Let $\norm{\GG} = \mrm{max}_{\vec u\in \Sigma} \norm{\vec u}_{1}$ be the largest $l_1$-norm of terminals. We define: \[\size{\GG} \defeq |\GG|{\cdot} \norm{\GG}.\]

\subsection{Thinness and Index}
In \cite{DBLP:conf/stoc/BiziereC25_BC25}, Bizi{\`{e}}re et al. distinguish GVAS by whether it is thin or branching. A GVAS is \emph{thin} if no nonterminal derives more than one copy of itself (i.e., there exists no derivation like $X\Rightarrow \alpha X \beta X \gamma$). The \emph{index} of CFG (similarly, GVAS) is introduced in \cite{DBLP:conf/fsttcs/AtigG11_AG11}. For a GVAS and a complete derivation $S_0\Rightarrow \alpha$, reordering the application of production rules gives different derivation sequences, like $S_0\rightarrow\alpha_1\rightarrow\alpha_2\rightarrow\dots\rightarrow\alpha$. A derivation sequence is \emph{optimal}, if it minimizes $\max\{\#_{\mrm{NT}}(\alpha_i)\}$, denoted by $k_{S_0\Rightarrow\alpha}$, where $\#_{\mrm{NT}}(\alpha_i)$ stands for number of nonterminals occurring in $\alpha_i$. The index of $\GG$ is denoted by $\iota(\GG)$, given by:
\[\iota(\GG) = \sup \left\{ k_{S_0\Rightarrow \alpha} \mid S_0\Rightarrow \alpha \text{ is optimal} \right\}.\]
Bizi{\`{e}}re et al. point out that every thin GVAS has a finite index \cite{DBLP:conf/stoc/BiziereC25_BC25}. In fact, a GVAS is thin iff it is finitely indexed. We are going to define thinness and the index in a graph-theoretic perspective. The equivalence of different definitions is proved in Appendix~\ref{section:appendix_GVAS}.

\emph{The production graph.} The production graph $\PG(V,E)$ of $\GG$ is a directed graph with $V = \Chi\cup \Sigma$. Every production rule of $X\rightarrow AB$ is mapped to a pair of edges $(X,A),(X,B)\in E$. Every rule of $X\rightarrow \Vec{u}$ is mapped to the edge $(X,\Vec{u})\in E$. Figure~\ref{fig:index_of_GVAS} illustrates the production graph. Consider the strongly-connected components (SCCs) in $\PG$. For the production rule $X\rightarrow AB$, one of the following holds:

\begin{enumerate}
    \item $X,A,B$ are in the same SCC. The rule is called \emph{nondegenerate}.
    \item Exactly one of $A,B$ is in the SCC of $X$. The rule is \emph{partially degenerate}. If $B$ shares the same SCC with $X$, the rule is \emph{left-degenerate}; it is \emph{right-degenerate} otherwise.
    \item Neither $A$ nor $B$ is in the SCC of $X$. The rule is \emph{fully degenerate}.
\end{enumerate}
If there exists any nondegenerate rule, every nonterminal in that SCC is capable of deriving two copies of itself, which violates the thinness. The production graph in Figure~\ref{fig:index_of_GVAS} contains no nondegenerate rules.

\begin{restatable}[Thinness]{definition}{DefinitionThinness}
    A GVAS $\GG$ is thin if $\Rho$ contains no nondegenerate production rules.
\end{restatable}

We abbreviate thin GVAS as TGVAS, and focus on TGVAS in the rest of the paper. For TGVAS $\GG$, we have two observations about the index based on $\PG$:

\begin{enumerate}
    \item For nonterminals $X,Y$ in the same SCC, the indices of their induced sub-GVASes are equal, since $X$ derives $Y$ and $Y$ derives $X$. That is, $\iota(\GG_X) = \iota(\GG_Y)$. For instance, there are three nontrivial SCCs $\{X,Y,Z\}, \{C,D\},  \{E,F\}$ in Figure~\ref{fig:index_of_GVAS}. Nonterminals in the same SCC share the same index.

    \item The index $\iota(\GG)$, as well as the index $\iota(\GG_X)$ for any $X\in \Chi$, is finite and computable through dynamic programming. We compute indices from lower SCCs up to the topmost SCC. The crux is that for $X\rightarrow AB$, we always expand first the nonterminal in a lower SCC, or with a smaller index. For example, in Figure~\ref{fig:index_of_GVAS}, given the right-degenerate rule $Z\rightarrow XF$ with $\iota(F)= 2$, it is immediate that $\iota(Z)\geq 3$ since one more nonterminal is required for symbol $X$.
\end{enumerate}

Denote the set of nonterminals in the topmost SCC of $\PG$ as $\Chi_{\mathrm{top}}(\PG)$. We define the index recursively (see Figure~\ref{fig:index_of_GVAS}):

\begin{restatable}[Index]{definition}{DefinitionIndex}
    Given a TGVAS $\GG$. If $\GG$ is trivial, $\iota(\GG)=1$. Otherwise, $\iota(\GG)$ is the smallest $k\in \bbn_+$ rendering true the following statements for all $X\in \Chi_{\mathrm{top}}(\PG)$.
    \begin{itemize}
        \item If there exists a left- (or right-) degenerate rule $X\rightarrow AY$ (or $X\rightarrow YA$), then $k\geq \iota(\GG_A)+1$.
        \item If there exists a fully-degenerate rule $X\rightarrow AB$ and $\iota(\GG_A)\neq \iota(\GG_B)$, then $k\geq \mathrm{max}\{\iota(\GG_A), \iota(\GG_B)\}$. If $\iota(\GG_A) =  \iota(\GG_B)$ , then $k\geq \iota(\GG_A) + 1$.
    \end{itemize}
    \label{definition:index}
\end{restatable}

\begin{figure}[!t]
    \begin{minipage}{0.60\columnwidth}
        \centering
        \includegraphics[width=0.95\columnwidth]{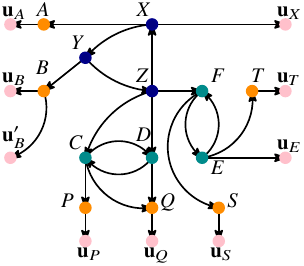}
    \end{minipage}
    \hfill
    \begin{minipage}{0.38\columnwidth}
        \centering
        The rule set $\Rho$ is:
        
        \vspace{2pt}
        \scriptsize
        %\tiny
        \begin{tabular}{c|c}
        $X\rightarrow AY$       & $A\rightarrow \vec u_A$      \\
        $Y\rightarrow BZ$       & $B\rightarrow \vec u_B$      \\
        $Z\rightarrow XF$       & $B\rightarrow \vec u_{B}'$   \\
        $Z\rightarrow CD$       & $P\rightarrow \vec u_P$      \\
        $C\rightarrow PD$       & $Q\rightarrow \vec u_Q$      \\
        $C\rightarrow PQ$       & $S\rightarrow \vec u_S$      \\
        $D\rightarrow CQ$       & $E\rightarrow \vec u_E$      \\
        $F\rightarrow SE$       & $T\rightarrow \vec u_T$      \\
        $E\rightarrow FT$       & $X\rightarrow \vec u_X$      \\
        \end{tabular}
        \includegraphics{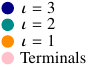}
    \end{minipage}
    \caption{The production graph and the index.}
    \Description{The production graph and the index of TGVAS $\GG$. The starting nonterminal is $S_0 = X$. There are three nontrivial SCCs: $\{X,Y,Z\}, \{C,D\},  \{E,F\}$. Nonterminals in the same SCC share the same index. There is no nondegenerate production rules, ensuring the thinness. We have $\iota(\GG) =\iota(\GG_X) =\iota(\GG_Y) =\iota(\GG_Z) = 3, \iota(\GG_C) =\iota(\GG_D)=2, \iota(\GG_E)=\iota(\GG_F) =2$ and so on.}
    \label{fig:index_of_GVAS}
\end{figure}

%\emph{Top:} The production graph and the index of TGVAS $\GG$. The starting nonterminal is $S_0 = X$. There are three nontrivial SCCs: $\{X,Y,Z\}, \{C,D\},  \{E,F\}$. Nonterminals in the same SCC share the same index. There is no nondegenerate production rules, ensuring the thinness. We have $\iota(\GG) =\iota(\GG_X) =\iota(\GG_Y) =\iota(\GG_Z) = 3, \iota(\GG_C) =\iota(\GG_D)=2, \iota(\GG_E)=\iota(\GG_F) =2$ and so on. \\
%\emph{Bottom:} Configurations of nonterminals in a complete derivation tree (left) and a segment (right). Given their root configurations, the computation orders differ. In the segment case, the configurations can also be obtained from the configuration $(1,4)$ in reverse order.

\subsection{Derivation Trees}
A derivation tree $\TT$ is a binary tree representing a derivation. Nodes in $\TT$ are denoted by lowercase letters $p,q,x,y,n,\dots$. Each node $p\in\TT$ is labeled with a symbol $\sigma(p) \in \Chi\cup\Sigma$. For every internal (non-leaf) node $p$, we have $\sigma(p)\in \Chi$, and one of the following cases applies:

\begin{figure}
    \centering
    \includegraphics[width=0.96\columnwidth]{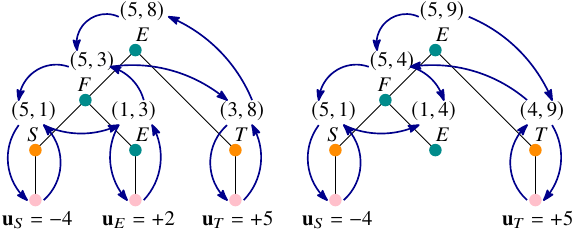}
    \caption{Configurations (terminal configurations omitted).}
    \Description{Configurations of nonterminals in a complete derivation tree (left) and a segment (right). Given their root configurations, the computation orders differ. In the segment case, the configurations can also be obtained from the configuration $(1,4)$ in reverse order.}
    \label{fig:conf}
\end{figure}

\begin{enumerate}
    \item $p$ has two children $c_1,c_2\in \TT$ and $\sigma(p)\rightarrow \sigma(c_1) \sigma(c_2) \in \Rho$. For simplicity, we refer to such situation as ``$p$ has the production rule $\sigma(p)\rightarrow \sigma(c_1)\sigma(c_2)$''.
    \item Or, $p$ has the rule $\sigma(p)\rightarrow \sigma(c_1)$ where $\sigma(c_1)\in \Sigma$ and $c_1$ is a leaf of $\TT$.
\end{enumerate}

The symbols of leaves in $\TT$ are not necessarily terminals. If they are all terminals, then $\TT$ is \emph{complete}. Otherwise, $\TT$ is \emph{incomplete}. We write $\TT_{X}$ to emphasize the root symbol being $X$. Given a derivation tree $\TT$, let $\Chi_{\TT},  \Sigma_{\TT}, \Rho_{\TT}$ denote the set of nonterminals, terminals and rules occurring in $\TT$. The subscript notation is extended to any reasonable sets or structures. The sum of all terminals in $\TT$ is the \emph{displacement} of $\TT$. Given two nodes $p,q\in\TT$. If $p$ is an ancestor of $q$, we denote $p\prec q$. The subtree rooted by node $p$ is $\sub{p}$.

\emph{Segments and cycles.} A \emph{segment} is an incomplete derivation tree with exactly one nonterminal leaf. Let the root be $p$ and the nonterminal leaf be $q$. The segment is denoted by $\seg{p,q}$. Given a complete derivation tree $\TT$ and $p,q\in\TT$ with $p\prec q$, the structure $(\sub{p} \setminus \sub{q}) \cup \{q\}$ is $\seg{p,q}$. If $\seg{p,q}$ satisfies $\sigma(p) = \sigma(q)$, it is a \emph{cycle}, denoted by $\cyc{p,q}$. A cycle is \emph{simple}, if it contains no strict subcycles. Intuitively, a cycle is the composition of simple cycles.

\emph{Configurations.} We define \emph{configurations} on complete derivation trees and segments. For every node $p$, it has a configuration tuple $(\vec l_p,\vec r_p)\in \bbn^d\times \bbn^d$. For a complete derivation tree $\TT$ rooted by $r$, $(\vec l_r, \vec r_r)$ is often fixed. Then, all configurations are determined as follows:

\begin{enumerate}
    \item Consider any internal node $p$. If $p$ has the rule $\sigma(p)\rightarrow \sigma(c_1)\sigma(c_2)$, then $\vec l_p = \vec l_{c_1}$, $\vec r_p = \vec r_{c_2}$ and $\vec r_{c_1} = \vec l_{c_2}$.
    \item If $p$ has the rule $\sigma(p)\rightarrow \sigma(c_1)$ and $\sigma(c_1) \in \Sigma$, then $\vec l_p = \vec l_{c_1}$, $\vec r_p = \vec r_{c_1}$ and $\vec l_{c_1} + \sigma(c_1) = \vec r_{c_1}$.
\end{enumerate}
The computation rules above also apply for segments, following a slightly different computational order illustrated in Figure~\ref{fig:conf}. Notice that for $\seg{x,y}$, fixing either $(\vec l_x,\vec r_x)$ or $(\vec l_y,\vec r_y)$ decides all other configurations.

A complete derivation tree $\TT$, or a segment $\seg{x,y}$ is \emph{nonnegative}, if all configurations are nonnegative, i.e., $(\vec l_p,\vec r_p)\geq (\vec 0,\vec 0)$ for every $p$. In the sequel, we adopt the \emph{most-simplified assumption} for derivation trees. That is, there do not exist nodes $p,q$ such that $p\prec q$, $\sigma(p)=\sigma(q)$ and $(\vec l_p,\vec r_p) = (\vec l_q, \vec r_q)$.

\subsection{Related Problems}

\emph{Reachability.} Given a $d$-TGVAS $\GG$, a pair of source and target vectors $(\vec s,\vec t)\in \bbn^d \times \bbn^d$, the reachability problem $\reach(\GG,\vec s,\vec t)$ asks whether there is a complete, nonnegative derivation tree $\TT_{S_0}$ with root configuration $(\vec s, \vec t)$.

\emph{Coverability.} Given $\GG$ and $(\vec s,\vec t)\in \bbn^d \times \bbn^d$, the coverability problem $\cover(\GG,\vec s,\vec t)$ asks whether there is a  complete, nonnegative derivation tree $\TT_{S_0}$ with root configuration $(\vec{s}, \vec t^\ast)$ with $\vec t^\ast \geq \vec t$.

The known complexities are stated in table \ref{tab:comparison}. The index is a key criterion to measure the hardness. We remark that the reachability for $k$-indexed $d$-TGVAS follows from the relation-extended VASS model proposed in~\cite{DBLP:conf/lics/GuttenbergCL25_CmeVASS} (See Appendix~\ref{section:appendix_GVAS} for the reduction). Using this model, the reachability for $k$-indexed 1-TGVAS is in $\FF_{6k-4}$. Moreover, our work relies on the fact that 1-GVAS coverability is in \EXPSPACE~\cite{DBLP:conf/icalp/LerouxST15_Coverability}. Our improved upper bound $\FF_{2k}$ is given by Theorem \ref{theorem:2k}, by generalizing the decomposition method for VASS~\cite{DBLP:conf/icalp/FuYZ24_F_d_KLMST,DBLP:conf/lics/LerouxS19_KLMST} to tree-shaped 1-TGVAS.

\begin{table}
\centering
\caption{Problems and their complexity}
\label{tab:comparison}

\begin{threeparttable}
\renewcommand{\arraystretch}{1.2}

\begin{tabular}{
    >{\centering\arraybackslash}p{2.8cm}
    >{\centering\arraybackslash}p{2.2cm}
    >{\centering\arraybackslash}p{2.2cm}
}
\toprule
\textbf{Model} & \textbf{Problem} & \textbf{Complexity}\\
\midrule
$k$-indexed 1-TGVAS & Reachability & $\FF_{2k}$ \tnote{a} \\
$k$-indexed $d$-TGVAS & Reachability & $\FF_{4kd+2k-4d}$ \cite{DBLP:conf/lics/GuttenbergCL25_CmeVASS} \\
$d$-GVAS & Reachability & Decidable \cite{DBLP:journals/corr/abs-2504-05015_PVASS_decidable} \\
1-GVAS & Coverability & \EXPSPACE \ \cite{DBLP:conf/icalp/LerouxST15_Coverability}\\
$d$-GVAS, $d\geq 2$ & Coverability & Decidable \tnote{b}\\

\bottomrule
\end{tabular}

\begin{tablenotes}
\item[a] The main result of this paper.
\item[b] By reducing $(d-1)$-reachability to $d$-coverability, the coverability problem inherits essentially the same hardness as the reachability.

\end{tablenotes}
\end{threeparttable}
\end{table}

\section{KLM Trees}
\label{section:KLM_tree}
The KLM decomposition algorithm for VASS \cite{DBLP:conf/icalp/FuYZ24_F_d_KLMST,DBLP:conf/lics/LerouxS19_KLMST} is based on KLM sequences, whose components encode the SCCs of VASS. Similarly, the components of a KLM tree capture the strong connectivity in derivation trees of TGVAS. We define the KLM tree for dimension $d$, in order to provide a generalization of the $\bbz$-reachability of $d$-TGVAS.

\emph{Strongly-connected segments.} Given $\seg{p,q}$ of a complete derivation tree, let $\ppath{p,q}$ be the path from $p$ to $q$. The production graph $\PG_{\seg{p,q}}$ is constructed according to rules used in $\seg{p,q}$. Let the topmost SCC of $\PG_{\seg{p,q}}$ be $\Chi_{\mrm{top}}(\PG_{\seg{p,q}})$. We define the strong connectivity on segments by:

\begin{definition}
    A segment $\seg{p,q}$ is \emph{strongly connected}, if for $n\in \ppath{p,q}$, it holds $\sigma(n)\in \Chi_{\mrm{top}}(\PG_{\seg{p,q}})$. That is, all nonterminals in $\ppath{p,q}$ lie in the same SCC.
\end{definition}

\subsection{Divisions}

Given a complete derivation tree $\TT$ rooted by $r$. The following procedure produces a {\em division} of $\TT$.

\begin{enumerate}
    \item If $\TT$ only contains a terminal, the procedure terminates.
    \item Find some $p\in \TT$ where $\seg{r,p}$ is strongly connected.
    \item Store $\seg{r,p}$ and do division on the children of $p$.
\end{enumerate}

Let all the segments so obtained be the set $\subdiv{\TT}$. There always exists a division since all segments can be trivial (containing only one node). Lemma \ref{lemma:small_division} implies the existence of a small division whose size is bounded by $\exp(\size{\GG})$. Before that, we introduce \emph{the topmost SCC division}. Given a complete, nontrivial $\TT$ rooted by $r$, there exists node $p\in \TT$ satisfying: 
\begin{enumerate}
    \item Segment $\seg{r,p}$ is strongly connected.
    \item Nonterminals in $\ppath{r,p}$ are \emph{unique}. For $n_1,n_2\in \TT$, if $n_1 \in \ppath{r,p}$ and $n_2 \notin \ppath{r,p}$, then $\sigma(n_1)\neq \sigma(n_2)$.
\end{enumerate}
To achieve this, consider the production graph $\PG_{\TT}$ constructed by rules used in $\TT$. Let the topmost SCC be $\Chi_{\mrm{top}}(\PG_{\TT})$. We initialize the node pointer $p$ as $r$, and check the rules on $\TT$ iteratively:
\begin{enumerate}
    \item If $p$ has a fully-degenerate rule on $\TT$, or $p$ has a terminal child, then terminate.
    \item If $p$ has a partially-degenerate rule $\sigma({p})\rightarrow \sigma(c_1)\sigma(c_2)$, update ${p}$ by $c_1$ if $\sigma(c_1)\in \Chi_{\mrm{top}}(\PG_{\TT})$ and update ${p}$ by $c_2$ if $\sigma(c_2)\in \Chi_{\mrm{top}}(\PG_{\TT})$. Go back to step (1).
\end{enumerate}
The node $p$ obtained at termination is the desired one. Strong connectivity follows directly from the definition of $\Chi_{\mrm{top}}(\PG_{\TT})$ and uniqueness is guaranteed by the definition of fully-degenerate rules.

\begin{lemma}
    \label{lemma:small_division}
    Given a TGVAS $\GG$ and a complete derivation tree $\TT$. There exists a division $\subdiv{\TT}$ with $|\subdiv{\TT}| \leq \EXP({\size{\GG}})$.
\end{lemma}

\begin{proof}
    Applying the topmost SCC division on $\TT$ yields the first segment $\seg{r,p}$. If $p$ has two children $c_1,c_2$, divide $\sub{c_1}$ and $\sub{c_2}$ recursively. By uniqueness,  $\Chi_{\sub{c_1}}, \Chi_{\sub{c_2}} \subsetneqq \Chi_{\TT}$ and the recursion depth has a limit of $|\Chi_{\TT}|$. We have $|\subdiv{\TT}|\leq 2^{|\Chi_{\TT}|} \leq \EXP({\size{\GG}})$.
\end{proof}

\subsection{KLM Components and KLM Trees}

Thanks to the existence of a small division, it suffices to depict the derivation behavior of strongly-connected segments. We next introduce the KLM component and the KLM tree, and explain how they faithfully capture derivation trees. The KLM components and KLM trees are syntactic objects given by the TGVAS grammar, whereas the segments and derivation trees are semantic objects.

\begin{definition}
    A KLM component is the tuple: \[\CC = \left(P,Q,\RhoSCC,\RhoL,\RhoR, \rhoD \right).\] Nonterminals $P,Q\in \Chi$ are the source and target symbols, respectively. $\RhoSCC,\RhoL,\RhoR \subseteq \Rho$ are three sets of production rules. $\rhoD \in \Rho$ is a single production rule.
\end{definition}

We are only interested in the KLM component that captures a strongly-connected segment. Figure~\ref{fig:KLM_components} shows how a strongly-connected segment (left, denoted by $\seg{p,q}$, and nonterminals $C,D$ not included) is captured by a KLM component (right). The nonterminals along $\ppath{p,q}$ form the sequence: $\{X,Y,Z,X,Y,Z\}$. The exit rule is $Z\rightarrow CD$. The idea of ``capturing'' is straightforward: let the rule sets $\RhoSCC, \RhoL,\RhoR$ collect production rules occurring in $\ppath{p,q}$, on its left, and on its right, respectively. The definition is as below.

\begin{definition}
    Given the KLM component $\CC$ and a strongly-connected segment $\seg{p,q}$. $\CC$ captures $\seg{p,q}$ with the exit rule $\sigma(q) \rightarrow \sigma(c_1)\sigma(c_2)$ if the following statements are valid. 
    \begin{enumerate}
        \item $P=\sigma(p), Q=\sigma(q)$ and $\rhoD=\sigma(q) \rightarrow \sigma(c_1)\sigma(c_2)$.
        \item $\RhoSCC$ is the set of production rules used along $\ppath{p,q}$.
        \item Let $N_{\mrm{L}}, N_{\mrm{R}}$ be two sets of tree nodes. For any parent node $n\in \ppath{p,q}\setminus \{q\}$, its left child $c_1$ and right child $c_2$, if $c_2\in \ppath{p,q}$, let $c_1\in N_{\mrm{L}}$; otherwise, let $c_2 \in N_{\mrm{R}}$. The rule sets $\RhoL,\RhoR$ satisfy the following conditions: 
        \[\RhoL = \bigcup_{n \in N_{\mrm{L}}} \Rho_{\sub{n}}, \RhoR = \bigcup_{n \in N_{\mrm{R}}} \Rho_{\sub{n}}. \]
    \end{enumerate}
    
    We write $\CC \capture \seg{p,q}$ for the capturing relation, omitting the exit rule. Usually $\rhoD$ is clear from context or is immaterial.
\end{definition}

A KLM component $\CC$ is trivial, if $\RhoSCC = \varnothing$. Hence, $\rhoD$ is the only rule in the trivial $\CC$. Observe that $\RhoSCC \cap \RhoL = \RhoSCC \cap \RhoR = \varnothing$ by thinness. The rule sets $\RhoL$ and $\RhoR$ may share identical production rules. In this case, they are treated as different. Nonterminals in $\ppath{p,q},\RhoL,\RhoR$ are denoted by $\ChiSCC, \ChiL, \ChiR$ respectively. Terminals in $\RhoL$ or $\RhoR$ are denoted by $\SigmaL$ or $\SigmaR$. For illustration, in Figure~\ref{fig:KLM_components}, $\ChiSCC = \{X,Y,Z\}$ and $\SigmaL = \{\vec u_A, \vec u_B\}$. In order to distinguish sets like $\ChiSCC$ or $\SigmaL$ belonging to different KLM components, we write $\ChiSCC(\CC)$ or $\SigmaL(\CC)$. 

\begin{figure}[!t]
    \begin{minipage}{0.57\columnwidth}
        \centering
        \includegraphics[width=0.95\columnwidth]{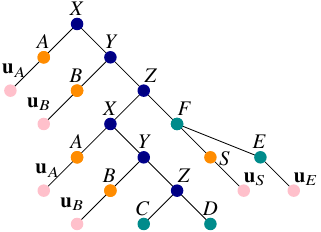}
    \end{minipage}
    \begin{minipage}{0.42\columnwidth}
        \centering
        \includegraphics[width=0.95\columnwidth]{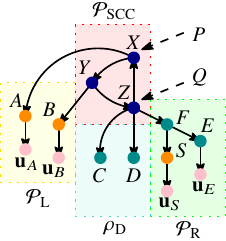}
    \end{minipage}
    \caption{A strongly-connected segment captured by a KLM component.} 
    \label{fig:KLM_components}
\end{figure}

%\emph{Left:} A strongly-connected segment $\seg{p,q}$ under the rule of example \ref{example:derivation_rule}. $\ppath{p,q}$ contains $X,Y,Z,X,Y,Z$ and the exit rule is $Z\rightarrow CD$.
%\emph{Right:} The KLM component $\CC$ capturing the segment. $\RhoSCC$ contains only partially-degenerate rules in $\ppath{p,q}$. $\RhoL$ and $\RhoR$ captures rules on the left or the right side of $\ppath{p,q}$. $\CC$ is similar to the production graph of $\seg{p,q}$. The key difference is that $\RhoL$ and $\RhoR$ are allowed to overlap. For example, if the rule $A\rightarrow \vec u_A$ also appears in the subtree of symbol $F$, it is contained in both $\RhoL$ and $\RhoR$.
%\begin{example}
%    Figure \ref{fig:KLM_components} illustrates a strongly-connected segment $\seg{p,q}$ and the KLM component $\CC\capture \seg{p,q}$. In this example, we have $\ChiSCC = \{X,Y,Z\}$ and $\SigmaL = \{\vec u_A, \vec u_B\}$. A KLM tree is organized by a collection of such capturing relations with a division.
%    \label{example:KLM_component}
%\end{example}

\begin{definition}
    A KLM tree $\KT$ is a binary tree whose nodes are KLM components or terminals. For $\CC\in \KT$, if $\rhoD = Q\rightarrow AB$, $\CC$ has two children $\CC_1,\CC_2$ whose sources $P(\CC_1)=A$ and $P(\CC_2)=B$. If $\rhoD = Q\rightarrow \vec u_{Q}$, $\CC$ has one child $\vec u_{Q}$. Terminal leaves are considered distinct, even if they have the identical values.
\end{definition}

\begin{definition}
    A KLM tree $\KT$ captures a complete derivation tree $\TT$, if there exists a division $\subdiv{\TT}$ and a bijection (maintaining the tree order) from KLM components $\CC\in \KT$ to segments $\seg{p,q}\in \subdiv{\TT}$ such that $\CC \capture \seg{p,q}$.
    The capturing relation is denoted by $\KT\capture \TT$.
\end{definition}

\subsection{The Characteristic System}
The syntax of a KLM component defines a class of possible, strongly-connected segments. To provide a more detailed characterization, we introduce an integer linear programming (ILP) system for each KLM component, known as the \emph{characteristic system}.  The characteristic system of $\CC$ contains following variables:
\begin{enumerate}
    \item Configuration variables. A KLM component ${\CC}$ has four configuration variables $\lPC,\rPC,\lQC,\rQC \in \bbn^d$, standing for the configurations at node $p$ or $q$ with respect to the captured segment $\seg{p,q}$.
    \item Variables of Parikh images of production rules. For each rule $\rho \in \RhoSCC$, $\RhoL$ or $\RhoR$, the variable $\#(\rho)\in \bbn_+$ denotes the Parikh image (how many times it occurs). Notice that identical rules in $\RhoL$ and $\RhoR$ are treated distinctly.
\end{enumerate}

Intuitively, a nonterminal must be produced and consumed the same number of times in order to form a complete derivation tree. Therefore, we introduce the implicit \emph{degree variables}. Given $X\in \ChiSCC$, let $\Rho_{\mrm{in}}(X)\subseteq \RhoSCC$ be the set of rules producing $X$, and $\Rho_{\mrm{out}}(X) \subseteq \RhoSCC$ be the set of rules consuming $X$. We define the degree variables as:
\begin{equation*}
    \begin{aligned}
        \indeg{X} &\defeq& \sum_{\rho\in \Rho_{\mrm{in}}(X)} \#(\rho), \\
        \outdeg{X} &\defeq& \sum_{\rho\in \Rho_{\mrm{out}}(X)} \#(\rho).
    \end{aligned}
\end{equation*}
Given $A \in \ChiL$, we define $\Rho_{\mrm{in}}(A)\subseteq (\RhoL\cup \RhoSCC)$. It covers rules in $\RhoL$ that produce $A$, and rules in $\RhoSCC$ that produce $A$ on the left. For example, $\Rho_{\mrm{in}}(A)$ contains rules $B\rightarrow AA\in \RhoL$ and $X\rightarrow AY\in \RhoSCC$, but not the rule $X\rightarrow YA\in \RhoSCC$. The variables $\indeg{A},\outdeg{A}$ are defined analogously, except that $\indeg{A}$ counts the Parikh image of rules like $B\rightarrow AA$ twice.
For terminals, the in-degrees are equivalent to the Parikh images. Therefore, given a terminal $\vec u \in \SigmaL$ and the rule set $\Rho_{\mrm{in}}(\vec u)\subseteq \RhoL$ producing $\vec u$, we define the Parikh image variable as:
\[\#(\vec u) \defeq \sum_{\rho\in \Rho_{\mrm{in}}(\vec u)} \#(\rho).\]

The characteristic system for a KLM component $\CC$ is denoted by $\EE_{\CC}$, containing the following equations:

\begin{equation*}
\begin{aligned}
        \#(\rho) &\geq 1, \\
        \lPC,\rPC,\lQC,\rQC &\geq \vec 0, \\
        \lPC + \sum_{\vec u \in \SigmaL} \#(\vec u)\cdot \vec u &= \lQC, \\
        \rQC + \sum_{\vec u \in \SigmaR} \#(\vec u)\cdot \vec u &= \rPC, \\
        \indeg{X} + \mathds{I} _{X = P} - \left(\outdeg{X} + \mathds{I}_{X = Q}\right)  &= 0,
\end{aligned}
\end{equation*}
where the logic indicator $\mathds{I}_{\Psi}=1$ if $\Psi$ is true and $\mathds{I}_{\Psi}=0$ if not.

The characteristic system $\EE_{\vec u}$ for a terminal $\vec u$ on the KLM tree contains variables $\vec l^{\vec u},\vec r^{\vec u} \in \bbn^d$ and is defined as:
\begin{equation*}\begin{aligned}
        \vec l^{\vec u}, \vec r^{\vec u} &\geq \vec 0, \\
        \vec l^{\vec u} + \vec u &= \vec r^{\vec u}.
\end{aligned}
\end{equation*}

Basically, the characteristic system of a KLM tree is composed of ILP systems like $\EE_{\CC}$ and $\EE_{\vec u}$. Moreover, the configuration variables should be subject to a global characteristic system in order to guarantee the consistency. The global characteristic system $\EE_{\mrm{Global}}$ is defined as follows. For any $\CC\in \KT$, if it has two children $\CC_1,\CC_2$, we add the equation:
\begin{equation*}
        \left(\lQC, \rQC, \rP^{\CC_1}\right) = \left(\lP^{\CC_1},\rP^{\CC_2},\lP^{\CC_2}\right).
\end{equation*}
Otherwise, it has a terminal child $\vec u$. We write:
\begin{equation*}
        \left(\lQC,\rQC\right) = \left(\vec{l}^{\vec{u}},\vec{r}^{\vec{u}}\right).
\end{equation*}

In some cases, the configuration variables of KLM components (or terminals) are fixed. The characteristic system $\EE_{\mrm{Config}}$ is introduced for such configuration constraints. Considering the reachability problem $\reach(\GG,\vec s,\vec t)$, we have the equation:
\begin{equation*}
    \left(\lP^{\CC}, \rP^{\CC}\right) = \left(\vec s, \vec t\right).
\end{equation*}
The characteristic system for a KLM tree is the union of ILP systems introduced above.

\begin{definition}
    The characteristic system for KLM tree $\KT$ is given by:
    \begin{equation*}\begin{aligned}
        \EE_{\KT} \defeq \left(\bigwedge_{\CC\in \KT} \EE_{\CC}\right) \land \left(\bigwedge_{\vec u\in \KT} \EE_{\vec u}\right) \land \EE_{\mrm{Global}} \land \EE_{\mrm{Config}}.
    \end{aligned}\end{equation*}
\end{definition}

\emph{The size of KLM trees.} The size of a KLM tree $\KT$ also serves as an upper bound for the unary size of $\KT$ and $\EE_\KT$. Let $|\CC|$ be the total number of symbols, rules, variables and equations in all characteristic systems involving $\CC$. Let $\norm{\CC}$ be the largest $l_1$-norm in the systems above. We define $\size{\CC} \defeq |\CC| \cdot \norm{\CC}$. 
Moreover, we define the size of $\KT$ as: 
\[\size{\KT} \defeq \sum_{\CC\in \KT}\size{\CC} + \sum_{\vec u\in \KT} \norm{\vec u}_1.\]

\emph{The initial capturing tree lemma} (Lemma \ref{lemma:initial_capturing_tree}) follows immediately from Lemma \ref{lemma:small_division}, together with the construction of the characteristic system.

\begin{lemma}
    \label{lemma:initial_capturing_tree}
    Given the reachability problem $\reach(\GG,\vec s,\vec t)$ and a instance of a complete, nonnegative derivation tree $\TT$, there exists a small KLM tree $\KT$ such that:
    
    \begin{enumerate}
        \item $\KT \capture \TT,$
        \item $\size{\KT} \leq \EXP(\size{\GG} + \norm{\vec s}_1 + \norm{\vec t}_1).$
    \end{enumerate}
\end{lemma}

\subsection{Ranks}
First, we introduce the \emph{geometric dimension} \cite{DBLP:conf/icalp/FuYZ24_F_d_KLMST,DBLP:conf/concur/Zheng25_geo_2_d,DBLP:conf/lics/GuttenbergCL25_CmeVASS} of KLM components. We may investigate cycles under the rules of a KLM component $\CC$. A \emph{top cycle} under $\CC$, denoted by $\cyc{x,y}$, is a cycle constructed within the syntax of $\CC$ that additionally satisfies $\sigma(x)=\sigma(y) \in \ChiSCC$. 

Let the symbols of the leaves of $\cyc{x,y}$ under DFS traversal (preorder) be the sequence: $\{\vec u_1, \dots, \vec u_{m_1}, \sigma(y), \vec u_{m_1+1}, \dots ,\vec u_{m_2}\}$. The effect of the cycle is defined as a vector in $\bbn^d \times \bbn^d$:
\[\Delta (\cyc{x,y}) \defeq \left(\sum_{i=1}^{m_1} \vec u_i, \sum_{i=m_1+1}^{m_2} \vec u_i\right).\]
Let $ {\mathcal{O}}_{\CC}$ be the set of all possible top cycles under $\CC$. We define the set of effects in ${\mathcal{O}}_{\CC}$ as:
\[\Delta(\mathcal{O}_\CC) \defeq \left\{\Delta(\cyc{x,y}) \mid \cyc{x,y}\in \mathcal{O}_\CC\right\}.\]
The geometric dimension of a nontrivial KLM component $\CC$ is denoted by $\gdim{\CC}$, defined as the dimension of the vector space spanned by the effects of possible top cycles under $\CC$. In formal terms, we define:
\[\gdim{\CC} \defeq \gdim{\Delta( {\mathcal{O}}_{\CC})}.\]
An observation is that $\gdim{\CC}\neq 0$ as long as $\CC$ is nontrivial. Otherwise, there must exist $n_1,n_2 \in \ppath{p,q}$ satisfying $n_1 \prec n_2, \sigma(n_1)=\sigma(n_2)$ and $\Delta(\cyc{n_1,n_2}) = (\vec 0, \vec 0)$, which violates the most-simplified assumption. We conclude that $\gdim{\CC} \in [1,2d]$ for any nontrivial $\CC$.

In a further scenario where a KLM component is decomposed, smaller KLM components may be extracted from the rule sets $\RhoL$ and $\RhoR$. We now introduce \emph{the index of KLM components}, which quantifies how many times such extractions may occur. The index of a nontrivial $\CC$ is defined as the largest index among induced sub-GVASes under $\RhoL$ or $\RhoR$ (instead of $\GG$):
\[\iota(\CC)\defeq \max_{X\in \ChiL \lor X\in \ChiR} \{\iota(\GG_{X})\}.\]
Here we slightly abuse notation by writing $\GG_X$ for the sub-GVAS under $\RhoL$ if $X\in \ChiL$, or under $\RhoR$ if $X\in \ChiR$. If $\GG$ is $k$-indexed, then $\iota(\CC) \in [1,k-1]$ for any nontrivial $\CC$.

\begin{definition}[Rank of KLM components]
    The rank of a nontrivial $\CC$ is the pair: \[\rank{\CC}\defeq \left(\iota(\CC), \gdim{\CC}\right).\]
\end{definition}
    
There are $(2kd-2d)$ different ranks. We compare $\rank{\mathcal{C}}$ under lexicographic order. The rank of $\KT$ is a vector in $\bbn^{2kd-2d}$, where the $i$-th coordinate counts the number of KLM components with the $i$-th largest rank. For example, the largest rank is $(k-1,2d)$, the second largest is $(k-1,2d-1)$, the $(2d+1)$-st largest is $(k-2,2d)$, and so on.

\begin{definition}[Rank of KLM trees]
    The rank of $\KT$ is defined as:
    \[ \rank{\KT} \defeq \left(\left(\sum_{\rank{\mathcal{C}} = (k-1,2d)} 1\right), \dots ,\left(\sum_{\rank{\mathcal{C}} = (1,1)} 1\right)\right).\]
\end{definition}

We also compare $\rank{\KT}$ under lexicographic order. The rank of a KLM component is computable by enumerating simple cycles under its rule sets (to compute the geometric dimension), and by dynamic programming (to compute the index).

\subsection{Solutions and Boundedness}
Given $\KT\capture \TT$, the characteristic system $\EE_{\KT}$ is sound and admits solutions representing complete derivation trees. However, they are not necessarily nonnegative. Moreover, $\EE_{\KT}$ admits at least one solution $\sol_\TT$ standing for $\TT$. A solution is called \emph{minimal} if there is no other solution that is pointwise less than or equal to it. Let $\solmin$ be a minimal solution under $\sol_\TT$. By Lemma~\ref{lemma:pottier}, there exists a size bound for $\solmin$:
\[\norm{\solmin}_{\infty} \leq \exp(\size{\KT}).\]

The homogeneous version of $\EE_\KT$, denoted by $\EE_{\KT}^0$, is obtained by replacing all nonzero constants with $0$. More precisely, we replace the corresponding equations with the following ones:
\begin{description}
  \item[$\EE_\CC$:]
  \[
      \#(\rho) \geq 0,
      \mrm{deg}_{\mrm{in}}(X) - \mrm{deg}_{\mrm{out}}(X) = 0 .
  \]

  \item[$\EE_{\vec u}$:]
  \[
    \vec l_{\vec u} = \vec r_{\vec u}.
  \]

  \item[$\EE_{\mrm{Config}}$:]
  \[
    (\lPC, \rPC) = (\vec 0, \vec 0).
  \]
\end{description}
Let $H$ denote the set of all nonnegative, minimal solutions to $\EE_{\KT}^0$. The set $H$ is known as the \emph{Hilbert basis}. Solutions in $H$ are compositions of cycles, since they fulfill the Euler condition. It is clear that $\sol_{\TT} \in \solmin + \bbn{(H)}$. We define the sum of solutions in $H$ as:
\[\solh \defeq \sum_{\sol\in H} \sol.\]
By Lemma~\ref{lemma:pottier}, we have $\norm{\solh}_{\infty} \leq \exp(\size{\KT})$.

For a certain variable $v$, let $\sol(v)$ denote the value of $v$ in the solution $\sol$. The condition $\solh(v)=0$ means that for any $\sol\in \solmin + \bbn{(H)}$, we have $\sol(v) =\solmin(v)$. Such variable $v$ is called \emph{bounded}. 
A variable is \emph{unbounded} otherwise. Bounded variables have an upper bound of $\exp(\size\KT)$.

\section{Decomposing KLM Trees}
\label{section:decomposing_KLM_tree}
\begin{figure}[!t]
    \centering
    \includegraphics[width = 0.94\columnwidth]{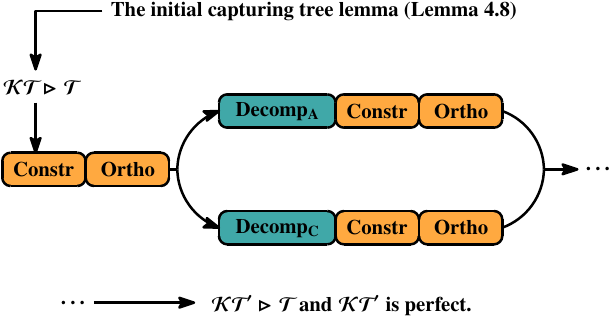}
    \caption{The refinement sequence.}
    \Description{The refinement sequence. Configuration constraining and orthogonalization are always performed collectively. The ``decomposition-cleaning'' loop brings about a strict decrease in rank.}
    \label{fig:refinement_seq}
\end{figure}

In this section, we focus on the one-dimensional case. We aim to convert a KLM tree into a \emph{perfect} KLM tree via a sequence of refinements,
where a perfect KLM tree is a certificate of the reachability. We introduce four refinements: configuration constraining, orthogonalization, algebraic decomposition and combinatorial decomposition, each of which guarantees a perfectness property.

Following Leroux et al.\ in \cite{DBLP:conf/lics/LerouxS19_KLMST}, we distinguish between two types of refinements: the \emph{cleaning operations}, including configuration constraining and orthogonalization, and the \emph{decompositions}, including the others. Applying any refinement on $\KT \capture \TT$ gives a new KLM tree $\NKT \capture\TT$ with $\size{\NKT}\leq \EXP(\size{\KT})$. A cleaning operation does not change the rank of $\KT$. A decomposition strictly decreases the rank, since it decomposes at least one KLM nontrivial component into KLM components with smaller ranks. The refinement sequence is shown in Figure \ref{fig:refinement_seq}. The length of the refinement sequence is limited by Lemma \ref{lemma:bad_seq}.
We refer to Appendix~\ref{section:appendix_decomposing_KLM_tree} for the detailed proofs of the lemmas sketched in the main text.

\subsection{Perfectness}
Given $\KT\capture\TT$, there exists a solution $\solt$ describing $\TT$. Nevertheless, the size of $\solt$ is unbounded, hence the reachability cannot be deduced solely from $\KT$. However, this limitation does not hold for a \emph{perfect} KLM tree. In section~\ref{section:perfect_KLM_tree}, we show that one can directly construct a nonnegative, complete derivation tree from a perfect KLM tree, without knowing the derivation tree it captures. Here, before the definition of the perfectness properties~\cite{DBLP:conf/lics/GuttenbergCL25_CmeVASS}, we introduce several concepts.

\emph{Orthogonality.} Generally, for a nontrivial $\CC\capture \seg{p,q}$, if all top cycles have effect $0$ on the left or right side, then configurations along $\ppath{p,q}$ on that side are entirely determined by their symbols. Such property is called orthogonality \cite{DBLP:conf/icalp/FuYZ24_F_d_KLMST} or rigidity \cite{DBLP:conf/lics/LerouxS19_KLMST}.
\begin{definition}
    Given a nontrivial $\CC\in \KT$. If for every $\cyc{x,y}\in  {\mathcal{O}}_{\CC}$, it holds that $\Delta(\cyc{x,y})[1] =0$, then $\CC$ is left-orthogonal; if $\Delta(\cyc{x,y})[2] =0$, it is right-orthogonal.
\end{definition}

\emph{Pumpability.} For a nontrivial KLM component $\CC$, the pumpability is discussed on configuration variables $\lPC,\lQC,\rPC,\rQC$. Let $\Pump$ be the set of configuration variables to be pumped. If $\lPC$ is bounded and $\CC$ is not left-orthogonal, then $\lPC \in \Pump$. If $\rPC$ is bounded and $\CC$ is not right-orthogonal, then $\rPC \in \Pump$, and analogously for $\lQC, \rQC$. 

Consider a nontrivial $\CC$ with at least one of $\lPC,\rPC$ in $\Pump$. $\CC$ is \emph{forward-pumpable}, if there exists a \emph{forward-pumping cycle} $\cyc{x,y}$ such that:
\begin{enumerate}
    \item $\cyc{x,y}\in \mathcal{O}_\CC$ and $\sigma(x) =\sigma(y) = P(\CC)$.
    \item If both $\lPC,\rPC\in \Pump$, fix the root configuration of $\cyc{x,y}$ by $(\vec l_x, \vec r_x) = (\lPC, \rPC)$. We require $\cyc{x,y}$ be nonnegative and $(\vec l_y, \vec r_y) \geq (\vec l_x+1, \vec l_y + 1)$.
    \item If $\lPC \in \Pump$ and $\rPC \notin \Pump$, fix the root configuration by $(\vec l_x, \vec r_x) = (\lPC, +\omega)$. We require $\cyc{x,y}$ be nonnegative and $\vec l_y \geq \vec l_x +1$. The symmetric case is handled similarly.
\end{enumerate}
The \emph{backward-pumpable} condition is symmetric. For instance, if $\lQC,\rQC\in \Pump$, we fix the leaf configuration by $(\vec l_y, \vec r_y) = (\lQC, \rQC)$ and require $(\vec l_x, \vec r_x) \geq (\vec l_y+1, \vec r_y+1)$. 

A component $\CC$ is \emph{pumpable}, if it is both forward-pumpable (when applicable, i.e., when $\Pump$ contains $\lPC$ or $\rPC$) and backward-pumpable. $\CC$ is called \emph{$f$-pumpable} for function $f(\cdot)$, if the pumping cycles $\cyc{x,y}$ satisfy $|\ppath{x,y}| \leq f(\size{\CC})$. 

We are ready to introduce the perfectness properties and the corresponding refinements. 

%For demonstration purposes, we consider only forward pumpability.

\begin{definition}[Perfectness]
A KLM tree $\KT \capture \TT$ is called \emph{perfect} if it satisfies the following four properties:

\begin{itemize}

\item \emph{Fully constrained.}
Every bounded configuration variable is fixed in the characteristic system
$\EE_{\mrm{Config}}$.

\item \emph{Fully orthogonalized.}
For every left-orthogonal component $\CC$ with bounded $\lPC$ and $\lQC$,
the left configurations are hard-encoded in $\ChiSCC$.
The right-orthogonal case is handled similarly.

\item \emph{Production unboundedness.}
For every nontrivial $\CC$ and every production rule $\rho \in \RhoSCC(\CC) \cup \RhoL(\CC) \cup \RhoR(\CC)$,
the variable $\#(\rho)$ is unbounded.

\item \emph{$\mrm{Exp}$-pumpability.}
Every nontrivial $\CC$ is $\exp$-pumpable.

\end{itemize}
\end{definition}

%\emph{Perfectness 1: fully-constrained property.} A KLM tree $\KT \capture \TT$ is fully constrained on configurations, if every bounded configuration variable is fixed in the characteristic system $\EE_{\mrm{Config}}$.

%\emph{Perfectness 2: fully-orthogonalized property.} A KLM tree $\KT\ \capture \TT$ is fully left-orthogonalized, if for every left-orthogonal $\CC$ with bounded $\lPC, \lQC$, the left configurations are hard-encoded in $\ChiSCC$. A similar argument applies to the right-orthogonal case.

%\emph{Perfectness 3: production unboundedness.} A KLM tree $\KT \capture \TT$ is unbounded on production rules, if for every $\CC$ and every $\rho\in \RhoSCC(\CC)$,$\RhoL(\CC)$ or $\RhoR(\CC)$, the variable $\#(\rho)$ is unbounded.

%\emph{Perfectness 4: $\exp$-pumpability.} A KLM tree $\KT \capture \TT$ is $\exp$-pumpable, if for every $\CC$, it is $\exp$-pumpable. 

\subsection{Configuration Constraining}
In order to meet the fully-constrained property, Lemma~\ref{lemma:constraining} is introduced. Applying Lemma~\ref{lemma:constraining} to $\KT$ is denoted by the refinement operator $\constr$. That is, $\constr(\KT) = \NKT$, where $\NKT$ is the new KLM tree obtained.

\begin{lemma}[Constraining]
    \label{lemma:constraining}
    Given $\KT\capture \TT$. There exists a fully-constrained KLM tree $\NKT$ such that:
    \begin{enumerate}
        \item $\NKT\capture \TT$,
        \item $\rank{\NKT} \leq \rank{\KT}$,
        \item $\size{\NKT} \leq \EXP(\size{\KT})$.
    \end{enumerate}
\end{lemma}

\begin{proof}
    Assume that in some component $\CC$, the variable $\lPC$ is bounded. Add the following constraint equation to $\EE_{\mrm{Config}}$: \[\lPC= \solt(\lPC).\] By boundedness, we have $\solt(\vec l_P^{\CC}) = \solmin(\vec l_P^{\CC})$, and by Lemma \ref{lemma:pottier}, the value is in $\EXP(\size{\KT})$. The size of the new KLM tree is also in $\exp(\size{\KT})$. Since we make no other modifications, the rank remains unchanged.
\end{proof}

\subsection{Orthogonalization}
Suppose the KLM tree is fully constrained. The orthogonalization is performed on orthogonal KLM components with bounded configuration variables. The refinement operator is $\ortho$.

\begin{lemma}[Orthogonalization]
    Given $\KT\capture \TT$ being fully constrained. There exists a fully-constrained, fully-orthogonalized KLM tree $\NKT$ such that:
    \begin{enumerate}
        \item $\NKT\capture \TT,$
        \item $\rank{\NKT} \leq \rank{\KT},$
        \item $\size{\NKT} \leq \EXP(\size{\KT}).$
    \end{enumerate}
    \label{lemma:orthogonalization}
\end{lemma}

\begin{restatable}{lemma}{LemmaDecidabilityOfOrthogonal}
    \label{lemma:decidability_of_orthogonal}
    The left-orthogonality of a KLM component $\CC$ (and the right-orthogonality as well) is decidable in \textup{\EXPSPACE}.
\end{restatable}

First, observe that for any left-orthogonal $\CC$ and any cycle under the rule set $\RhoL(\CC)$, the cycle has a displacement of $0$ (recall that the displacement of a cycle is the sum of terminals). Moreover, whenever a nonterminal (such as $A$) is produced directly by a left-degenerate rule in $\RhoSCC(\CC)$ (such as $X\rightarrow AY$), the displacements of all these subtrees of $A$ are identical. The displacement is called the \emph{uniform displacement}, denoted by $\vec u_A$. Clearly, $\vec u_A$ is bounded by $\exp(\size\CC)$ since it is the displacement of some acyclic, complete derivation tree under $\RhoL(\CC)$.

%To prove the lemma, we introduce {\color{red}\textbf{uniform displacement}}. 
%Given $A\in\ChiL$, let $\GG_A$ be the induced sub-GVAS of $A$ under the rule set $\RhoL$. In $\GG_A$, every complete derivation tree $\TT_A$ must have the same displacement. Such effect can be achieved by introducing a pseudo-rule $A \rightarrow \vec u_A$ where $\vec u_A$ denotes the \emph{uniform displacement}. 

\begin{proof}[Sketch of Proof]
    First, we check whether such uniform displacements exist. Then, we construct a weighted, directed graph $G(\ChiSCC, E)$ by mapping left-degenerate rules like $X\rightarrow AY \in \RhoSCC(\CC)$ to the edge $(X,Y,\vec u_A)$, and right-degenerate rules like $X\rightarrow YB\in \RhoSCC(\CC)$ to the edge $(X,Y,0)$. The weight sum of any cycle in $G$ must be $0$. All checks can be done in \EXPSPACE.
\end{proof}

If $\CC\capture \seg{p,q}$ is left-orthogonal, $\lPC$ and $\lQC$ share the same boundedness. If both $\lPC$ and $\lQC$ are bounded, the left configuration of any $n\in \ppath{p,q}$ is also bounded by $\exp(\size\CC)$ and can be determined solely from $\sigma(n)$. To see this, consider the graph $G(\ChiSCC,E)$ constructed above. The left configuration is obtained through an acyclic path in $G$, which is adding at most $|\ChiSCC|$ uniform displacements to $\lPC$. This value, denoted by $\vec l_{\sigma(n)}$, is called the \emph{uniform configuration}.

For $X\in \ChiSCC$, the uniform left configuration $\vec l_X$ is effectively computable. We intend to replace symbol $X$ with the new symbol $(_{\vec l_X}X)$ where $\vec l_X$ is hard-encoded, such as $(_3X), (_5Y)$. For the right-orthogonal case, new symbols are $(X_3),(Y_5)$ and so on. Since the uniform configurations are hard-encoded, for the left-degenerate rule $(_3X)\rightarrow A(_5Y)$, the derivation tree $\TT_A$ must be a nonnegative, complete derivation tree with the root configuration $(3,5)$. The rule set $\RhoL(\CC)$ is no longer relevant. Our attention now shifts to whether such certificate exists. Therefore, rules in $\RhoL(\CC)$ are replaced with rules like $A\rightarrow 5-3$, namely the \emph{certificate rules}. Notice that $5-3=2$ is exactly the uniform displacement $\vec u_A$.

\iffalse
\textbf{Perfectness 2: fully-orthogonalized property.} A KLM tree $\KT\ \capture \TT$ is fully left-orthogonalized, if for every left-orthogonal $\CC\in \KT$ with bounded $\vec l_P^\CC,\vec l_Q^\CC$ the following holds:
\begin{enumerate}
    \item The left configuration $\vec l_X$ is recorded for $X\in \ChiSCC$. i.e., nonterminal $X$ is replaced by $(_{\vec l_X}X)$, such as $(_3X),(_5Y)$ for the one-dimensional case. 
    The starting nonterminal $P$, the targeting nonterminal $Q$ and nonterminals in the rules of $\RhoSCC$ are modified accordingly.
    \item The left rule set $\RhoL$ is cleared and reconstructed. For every left-degenerate rule $(_{\vec l_X}X)\rightarrow A(_{\vec l_Y}Y) \in \RhoSCC$, we have the rule $A \rightarrow \vec u_A\in \RhoL$ where $\vec u_A$ is the uniform displacement. It always holds that $\vec u_A = \vec l_Y - \vec l_X$.
\end{enumerate}

The same requirement applies for the right-orthogonal case. Rules in $\RhoL$ serve as oracles. For a rule $(_{\vec l_X}X) \rightarrow A (_{\vec l_Y} Y) \in \RhoSCC$, the rule $A \rightarrow (\vec l_Y -\vec l_X)\in \RhoL$ is a reachability oracle $\reach(\GG_A,\vec l_X, \vec l_Y)$ with $\iota(\GG_A)< k$. All those oracles are checked in our final reachability algorithm recursively.
\fi

\begin{proof}[Proof for Lemma~\ref{lemma:orthogonalization}]
    Consider a nontrivial $\CC\in \KT$ which is left-orthogonal. Suppose that $\lPC$ and $\lQC$ are bounded. The uniform displacements and the uniform left configurations are computable and bounded by $\exp(\size\CC)$. We reconstruct $\CC$ as follows:
    \begin{enumerate}
        \item For any rule $X\rightarrow AB \in \RhoSCC(\CC)$, replace it with $(_{\vec l_X}X)\rightarrow A(_{\vec l_B}B)$.
        \item Replace $P(\CC)$ with $(_{\vec l_P}P)$ and replace $Q(\CC)$ with $(_{\vec l_Q}Q)$.
        \item The rule set $\RhoL(\CC)$ is cleared. For every left-degenerate rule $(_{\vec l_X}X)\rightarrow A(_{\vec l_Y}Y)$, add the certificate rule $A\rightarrow \vec l_Y-\vec l_X$.
    \end{enumerate}
    The size of the new KLM component (including values like $\vec l_X$) is bounded by $\exp(\size\CC)$. The index does not increase, since rules in $\RhoL(\CC)$ are now simplified. The geometric dimension does not change. We conclude that $\NKT$ satisfies the requirement above.
\end{proof}

In the sequel, configuration issues on the orthogonalized side are no longer considered. The certificate rule $A\rightarrow \vec l_Y -\vec l_X$ serves as a reachability subproblem $\reach(\GG_A, \vec l_X,\vec l_Y)$ with $\iota(\GG_A) \leq \iota(\GG)-1$. In our final algorithm, such subproblems are checked recursively.  Moreover, the definition of $\capture$ is extended. For $\CC \capture \seg{p,q}$ where $\CC$ is left-orthogonalized, on top of all previous conditions, the condition involving $\RhoL$ is replaced by:
\begin{quote}
    For $n\in \ppath{p,q}$ such that $n$ has a left-degenerate rule $\sigma(n) \rightarrow \sigma(c_1)\sigma(c_2)$ on the segment, there exists a certificate rule in $\RhoL(\CC)$ representing $\sub{c_1}$.
\end{quote}

\subsection{Algebraic Decomposition}
\label{subsection:alg}

\begin{figure*}[!t]
    \centering
    \includegraphics[width=0.72\textwidth]{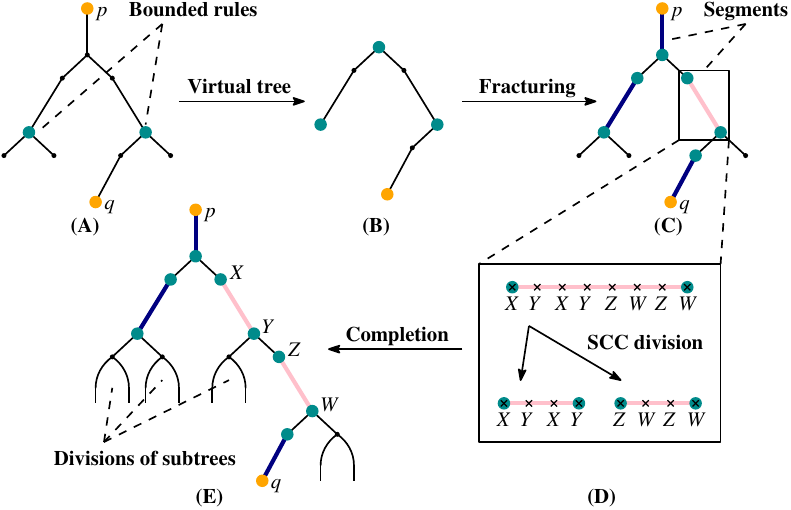}
    \caption{The division of bounded segments.}
    \Description{Read this figure clockwise.}
    \label{figure:alg_decomp}
\end{figure*}
The algebraic decomposition is performed on any nontrivial KLM component $\CC$ with any bounded $\#(\rho)$. The refinement operator is $\decompA$.

\begin{restatable}[Algebraic decomposition]{lemma}{LemmaAlgebraicDecomposition}
    \label{lemma:algebraic decomposition}
    Given $\KT \capture \TT$. If there exists some nontrivial $\CC \in \KT$ with at least one bounded $\#(\rho)$, then there exists another KLM tree $\NKT$ such that:
    \begin{enumerate}
        \item $\NKT\capture \TT,$
        \item $\rank{\NKT} < \rank{\KT},$
        \item $\size{\NKT} \leq \EXP(\size{\KT}).$
    \end{enumerate}
\end{restatable}

Let $\CC\capture \seg{p,q}$. We aim to decompose $\CC$ into smaller KLM components whose ranks strictly decrease. In order to maintain the capturing relation, we have to subdivide $\seg{p,q}$ into smaller segments. We construct a division set $\subdiv{\seg{p,q}}$ containing these segments which renders true the following statements:
\begin{enumerate}
    \item Bounded rules are excluded from segments in $\subdiv{\seg{p,q}}$.
    \item Every segment in $\subdiv{\seg{p,q}}$ is strongly connected.
    \item Let the original division set of $\TT$ be $\subdiv{\TT}$. Then: \[\subdivs{\TT}\defeq  (\subdiv{\TT} \setminus \{\seg{p,q}\}) \cup \subdiv{\seg{p,q}}\] is a new division of $\TT$.
    \item The size $|\subdiv{\seg{p,q}}|$ is bounded by $\exp(\size{\KT})$.
\end{enumerate}
Our main idea is illustrated in Figure~\ref{figure:alg_decomp}. 
Diagram~(A) shows $\seg{p,q}$ with two bounded rules, while diagrams~(B)–(E) illustrate the four steps of the division procedure.

\emph{Step 1: Virtual tree construction.}
In this step, we intend to locate the bounded rules in $\seg{p,q}$. Let the set of bounded rules be $\RhoB\subseteq \RhoSCC\cup \RhoL\cup \RhoR$. Let the node set $N_{\Bounded}$ collect all nodes in $\seg{p,q}$ that have a bounded rule. We additionally include the node $q$ in $N_{\Bounded}$ since $\rhoD$ is used exactly once. Denote the lowest common ancestor (LCA) of nodes $n_1,n_2$ as $\mrm{lca}(n_1,n_2)$. Consider the LCA-closure of $N_{\Bounded}$:
\[\mrm{LCA}(N_{\Bounded}) \defeq N_{\Bounded} \cup \left\{\mrm{lca}(n_1,n_2) \mid n_1,n_2\in N_{\Bounded}\right\}.\]
Nodes in $\mrm{LCA}(N_{\Bounded})$ form the \emph{virtual tree} of $N_{\Bounded}$. In Figure~\ref{figure:alg_decomp}, diagram~(B), the virtual tree contains four nodes: node $q$, two nodes with bounded rules, and their LCA.

\emph{Step 2: Segment fracturing.}
Every edge of the virtual tree is going to be a new segment. For every $q^\ast\in \mrm{LCA}(N_{\Bounded})$, we construct $\seg{p^\ast,q^\ast}$ and decide $p^\ast$ by:

\begin{enumerate}
    \item Let $r\in \mrm{LCA}(N_{\Bounded})$ be the root of the virtual tree. For $q^\ast\in \mrm{LCA}(N_{\Bounded})$ and $q^\ast \neq r$, it has a parent in the virtual tree, denoted by $\mrm{par}(q^\ast)$ .
    \item  Node $\mrm{par}(q^\ast)$ has two children $c_1,c_2$ on $\seg{p,q}$. Let $p^\ast$ be $c_1$ if $q^\ast \in \sub{c_1}$, and $c_2$ otherwise.
    \item For $q^\ast = r$, let $p^\ast = p$.
\end{enumerate}
In this way, bounded rules serve as bridges between segments and are therefore effectively excluded. See Figure~\ref{figure:alg_decomp}, diagram (C).

\emph{Step 3: SCC division.}
In this step, we further subdivide segments $\seg{p^\ast,q^\ast}$ into strongly-connected ones. We construct a directed graph $G$ by merging nodes on $\ppath{p^\ast,q^\ast}$ which share the same symbol. The SCC condensation graph of $G$ is a path. Each SCC is mapped to a subpath, denoted by $\ppath{p^\dagger, q^\dagger}$, and gives a strongly-connected segment $\seg{p^\dagger, q^\dagger}$. In Figure~\ref{figure:alg_decomp}, diagram~(D), the original path $\{X,Y,X,Y,Z,W,Z,W\}$ is divided into $\{X,Y,X,Y\}$ and $\{Z,W,Z,W\}$.

\emph{Step 4: Subtree completion.} 
In any feasible division of a complete derivation tree, each segment has either exactly two children or none. However, for a segment $\seg{p^\dagger,q^\dagger}$ obtained in step 3, the subtree rooted by a child of $q^\dagger$ is not necessarily covered in any segment. Subdivisions of these subtrees should be included in the final division set, as shown in Figure~\ref{figure:alg_decomp}, diagram (E).
\begin{enumerate}
    \item  Let $\subdiv{\seg{p,q}}$ collects all strongly-connected $\seg{p^\dagger, q^\dagger}$ first. We check all $\seg{p^\dagger, q^\dagger}$ as follows.
    \item The children of $q^\dagger$ must be covered. For an uncovered child $c_1$, we can infer that $c_1\notin \ppath{p,q}$. Therefore, $\sub{c_1}$ is a complete derivation tree.
    \item By Lemma \ref{lemma:small_division}, there exists a small division of $\sub{c_1}$. The size bound is $|\subdiv{\sub{c_1}}| \leq \exp(\size{\CC})$. We add all segments in $\subdiv{\sub{c_1}}$ to $\subdiv{\seg{p,q}}$.
\end{enumerate}
The division of a bounded segment is done. The new division of $\TT$ given by:
\[\subdivs{\TT}= (\subdiv{\TT} \setminus \{\seg{p,q}\}) \cup \subdiv{\seg{p,q}}\] is a feasible one. The size of $\subdiv{\seg{p,q}}$ is bounded by Lemma~\ref{lemma:size bound for segment division}, proved in Appendix~\ref{section:appendix_decomposing_KLM_tree}.
\begin{restatable}{lemma}{LemmaSizeofDivision}
    The division set $\subdiv{\seg{p,q}}$ for the algebraic decomposition has a size bound of $\exp(\size{\KT})$.
    \label{lemma:size bound for segment division}
\end{restatable}

\begin{proof}[Sketch of Proof for Lemma \ref{lemma:algebraic decomposition}]
    Given $\KT \capture \TT$. Suppose that $\CC$ contains at least one bounded $\#(\rho)$. Subdividing $\seg{p,q}$ gives $\subdiv{\seg{p,q}}$. For each segment $\seg{p^\ddagger, q^\ddagger}\in \subdiv{\seg{p,q}}$, we systematically construct a new KLM component $\CC^\ddagger \capture \seg{p^\ddagger, q^\ddagger}$ according to the definition of the capturing relation. We distinguish these new KLM components by two cases:
    \begin{enumerate}
        \item \emph{The overlapping case.} If $\ppath{p^\ddagger, q^\ddagger}\subseteq \ppath{p,q}$, the rule sets $\RhoSCC(\CC^\ddagger)$, $\RhoL(\CC^\ddagger)$ and $\RhoR(\CC^\ddagger)$ are the subsets of the counterparts in $\CC$.  If $p^\ddagger = p$ or $q^\ddagger =q$, $\CC^{\ddagger}$ inherits the configuration constraining equations in $\EE_{\mrm{Config}}$. We have $\size{\CC^\ddagger} = O(\size{\CC})$.
        \item \emph{The disjoint case.} Otherwise, $\ppath{p^\ddagger, q^\ddagger} \cap \ppath{p,q} = \varnothing$. The new KLM component $\CC^\ddagger$ contains only rules in $\RhoL(\CC)$ or $\RhoR(\CC)$. We also have $\size{\CC^\ddagger} = O(\size{\CC})$.
    \end{enumerate}
    By Lemma~\ref{lemma:size bound for segment division}, it is immediate that $\size{\NKT} \leq \exp(\size{\KT})$. 
    We analyze the decrease in the rank according to the two cases above. For the disjoint case, it immediately follows that $\iota(\CC^{\ddagger}) \leq  \iota(\CC^{\ddagger}) -1$. For the overlapping case, $\iota(\CC^{\ddagger}) \leq  \iota(\CC^{\ddagger})$, and we provide a sketch of $\gdim{\CC^{\ddagger}} < \gdim{\CC}$ by contradiction. Assume that $\gdim{\CC^{\ddagger}} = \gdim{\CC}$. We can construct a top cycle in $\mathcal{O}_{\CC}$ with at least one bounded rule, whose effect is in $\bbq(\Delta(\mathcal{O}_{\CC^\ddagger}))$. Therefore, the Parikh image of the bounded rule is represented as a linear combination of unbounded ones, yielding a contradiction. See Appendix~\ref{section:appendix_decomposing_KLM_tree} for a detailed proof.
\end{proof}

\subsection{Combinatorial Decomposition}
The combinatorial decomposition is performed on any nontrivial, fully-constrained and fully-orthogonalized KLM component $\CC$ with $\Pump\neq \varnothing$. The refinement operator is $\decompC$.

\begin{restatable}[Combinatorial decomposition]{lemma}{LemmaCombinatorialDecomposition}
    Given a fully-constrained and fully-orthogonalized KLM tree $\KT \capture \TT$. If $\KT$ is not $\exp$-pumpable, then there exists another KLM tree $\NKT$ such that
    \begin{enumerate}
        \item $\NKT\capture \TT,$
        \item $\rank{\NKT} < \rank{\KT},$
        \item $\size{\NKT} \leq \EXP(\size{\KT})$.
    \end{enumerate}
    \label{lemma:combinatorial decomposition}
\end{restatable}

Lemma~\ref{lemma:one sided pumping} and Lemma~\ref{lemma:two sided pumping} establish the following fact: if there exists a sufficiently large configuration, then $\exp$-pumpability holds. Actually, it is the contrapositives of Lemmas~\ref{lemma:one sided pumping} and \ref{lemma:two sided pumping} that are useful for the proof of Lemma~\ref{lemma:combinatorial decomposition}.

\begin{restatable}[One-sided forward-pumping]{lemma}{LemmaOneSidedPumping}
    Let $\KT$ be fully constrained. For a nontrivial $\CC \capture \seg{p,q}$, if $\lPC \in \Pump, \rPC \notin \Pump$, and there exists $n\in \ppath{p,q}$ with its left configuration $\vec l_n$ satisfying $\vec l_n \geq \exp(\size{\CC})$, then we can construct a forward-pumping cycle $\cyc{x,y}$ with length bound $|\ppath{x,y}|\leq \exp(\size{\CC})$.
    \label{lemma:one sided pumping}
\end{restatable}

\begin{restatable}[Two-sided forward-pumping]{lemma}{LemmaTwoSidedPumping}
    Let $\KT$ be fully constrained. For a nontrivial $\CC \capture \seg{p,q}$, if $\lPC,\rPC \in \Pump$, and there exist $n_i,n_j\in \ppath{p,q}$ such that $\vec l_{n_i} \geq \exp(\size{\CC})$ and $\vec r_{n_j} \geq \exp(\size{\CC})$, then we can construct a forward-pumping cycle $\cyc{x,y}$ with length bound $|\ppath{x,y}|\leq \exp(\size{\CC})$.
    \label{lemma:two sided pumping}
\end{restatable}

The main pumping technique for VASS is the \emph{Rackoff lemma}~\cite{rackoff}, based on induction. By appending a short pumping cycle for $(n-1)$ counters after a sufficiently large configuration, we obtain a pumping cycle for $n$ counters. However, this does not hold directly for GVAS. For $\cyc{x,y}$, being short (i.e., $|\ppath{x,y}|$ is small) does not imply that the effect is limited. Another choice is to discuss the pumpability on the relation-extended VASS model~\cite{DBLP:conf/lics/GuttenbergCL25_CmeVASS}, which fits well with the \emph{backwards coverability algorithm}, a straightforward yet computationally expensive brute-force algorithm.

For 1-TGVAS, we combine the techniques described above and exploit the \EXPSPACE{} upper bound for coverability. First, we introduce the \emph{forward-coverability upper bound function}:
\begin{definition}
    Given $A\in \ChiL$ and let $\GG_A$ be the induced sub-GVAS under $\RhoL$. The forward-coverability upper bound function $\delta_A(s) : \bbnww \rightarrow \bbnww$ is given by:
    \begin{enumerate}
        \item $\delta_A(-\omega) \defeq -\omega, \delta_A(+\omega) \defeq +\omega$.
        \item If $\cover(\GG_A, s, 0)$ rejects, then $\delta_A(s) \defeq -\omega$.
        \item Otherwise, $\delta_A(s) \defeq \sup\{t\in \bbn \mid \cover(\GG_A, s, t) \text{ accepts}\}$.
    \end{enumerate}
\end{definition}

Clearly, $\delta_A(s)$ is the largest value reachable from $s$. Notice that there exists $\Delta \leq \exp(\size\CC)$ such that for $s\geq \Delta$, $\delta_A(s)\geq s - \Delta$. That is true as long as $\Delta$ is greater than the absolute displacement of any acyclic, complete derivation tree. For $B\in \ChiR$, the upper bound function is defined on the \emph{mirrored} version of $\GG_B$, where $X\rightarrow CD$ is replaced by $X\rightarrow DC$ and $X\rightarrow \vec u$ is replace by $X\rightarrow -\vec u$.

Let the identical function be $\mrm{Id}:\bbn \rightarrow \bbn, \mrm{Id}(x)\defeq x$. We construct a directed graph $G(\ChiSCC, E)$ whose edges are labeled by relations on $\bbn^2\times \bbn^2$:

\begin{enumerate}
    \item For a left-degenerate rule $X\rightarrow AY\in \RhoSCC$, we add the edge $(X,Y)$ labeled by $(\delta_A(\cdot), \mrm{Id}(\cdot))$.
    \item For a right-degenerate rule $X\rightarrow YB\in \RhoSCC$, we add the edge $(X,Y)$ labeled by $(\mrm{Id}(\cdot), \delta_B(\cdot))$.
\end{enumerate}
Graph $G$ is actually a relation-extended 2-VASS defined in~\cite{DBLP:conf/lics/GuttenbergCL25_CmeVASS}. The pumpability problem for $\CC$ is equivalent to a coverability problem in $G$. To be specific, we consider the one-sided pumping case. The existence of an $\exp$-pumping cycle is equivalent to the existence of a cycle in $G$ satisfying:
\begin{enumerate}
    \item The cycle is from $P$ to $P$.
    \item The start configuration is $(\lPC, +\omega)$ and the final configuration is no less than $(\lPC+1, +\omega)$.
    \item The length is bounded by $\exp(\size\CC)$.
\end{enumerate}
By such conversion, it suffices to study the cycles in $G$. We give the sketched proof for the one-sided pumping lemma.

\begin{proof}[Sketch of Proof for Lemma \ref{lemma:one sided pumping}]
    If a path achieves a left configuration greater than $\lPC + |\ChiSCC|\cdot \Delta + 1$, then directly returning to $P$ yields a pumping cycle. By requiring $\vec l_n$ to be sufficiently large, and eliminating cycles on $\ppath{p,n}$, we obtain a pumping cycle whose length is bounded by $\exp(\size\CC)$.
\end{proof}

\begin{restatable}{lemma}{LemmaDecidabilityofPumping}
    For a nontrivial $\CC$ with $\Pump\neq \varnothing$, the $\exp$-pumpability is decidable in \textup{\EXPSPACE} under unary encoding.
    \label{lemma:pumpability is decidable}
\end{restatable}

\begin{proof}[Sketch of Proof]
    We intend to construct $G$ and enumerate cycles with length bound $\exp(\size\CC)$. Despite the fact that $\delta_A(s)$ may not be effectively computable, it is sufficient to compute:
    \[\max\left\{\delta_A(s), \lPC + \exp(\size\CC) \cdot \Delta + 1\right\},\]
    which is obtained by enumerating $t\leq \exp(\size\CC)$ and ask the coverability problem $\cover(\GG_A,s,t)$.
\end{proof}

Let $\KT \capture \TT$ be fully constrained and fully orthogonalized. Consider a nontrivial $\CC\capture\seg{p,q}$ that is not $\exp$-forward-pumpable. We can assume that $\forall n\in \ppath{p,q}, \vec l_n\leq v_{\mrm{B}}$ for $v_{\mrm{B}} = \exp(\size{\CC})$. We construct the division set $\subdiv{\seg{p,q}}$ first.

\emph{Step 1: Configuration encoding.} We construct an identical segment $\seg{p^\ast,q^\ast}$ except that bounded configurations are hard-encoded in nonterminals in $\ppath{p^\ast,q^\ast}$. New nonterminals are in $[0,v_{\mrm{B}}] \times \ChiSCC$. For $n^\ast\in \ppath{p^\ast, q^\ast}$ and its counterpart $n\in \ppath{p,q}$, we change the symbol by $\sigma(n^\ast) = (_{\vec l_n} \sigma(n))$, like $(_3X)$.

\emph{Step 2: SCC division.} Subdivide $\seg{p^\ast, q^\ast}$ by strong connectivity. The same as described in Subsection~\ref{subsection:alg}.

\emph{Step 3: Subtree completion.} Also the same. The size of the divisions set $\subdiv{\seg{p,q}}$ is bounded by Lemma~\ref{lemma:size bound for segment division two}.

%It holds that:
%\[\subdivs{\TT}\defeq  (\subdiv{\TT} \setminus \{\seg{p,q}\}) \cup \subdiv{\seg{p,q}}\] is a new division of $\TT$.

\begin{restatable}{lemma}{LemmaSizeofDivisionTwo}
    The division set $\subdiv{\seg{p,q}}$ for the combinatorial decomposition has a size bound of $\exp(\size{\CC})$.
    \label{lemma:size bound for segment division two}
\end{restatable}

\begin{proof}[Sketch of Proof for Lemma \ref{lemma:combinatorial decomposition}]
    Given a fully-constrained and fully-orthogonalized $\KT\capture \TT$. Let the $\CC\capture \seg{p,q}$ be nontrivial and not $\exp$-forward-pumpable. According to our assumption, we have $\forall n\in \ppath{p,q}, \vec l_n\leq v_{\mrm{B}}$. The division set is $\subdiv{\seg{p,q}}$. For $\seg{p^\ddagger,q^\ddagger}\in\subdiv{\seg{p,q}}$, we construct a new KLM component $\CC^\ddagger \capture \seg{p^\ddagger,q^\ddagger}$ and carry out the following case analysis:
    \begin{enumerate}
        \item \emph{The overlapping case.} If $\ppath{p^\ddagger,q^\ddagger}\subseteq \ppath{p,q}$, then we construct a left-orthogonalized $\CC^\ddagger$. Left configurations in $\ppath{p,q}$ are hard-encoded in nonterminals. For instance, for $n\in \ppath{p,q}$ which has a left-degenerate rule $\sigma(n)\rightarrow \sigma(c_1)\sigma(c_2)$, we add the rule $(_{\vec l_n}\sigma(n)) \rightarrow \sigma(c_1)(_{\vec l_{c_2}} \sigma(c_2))$ to $\RhoSCC(\CC^\ddagger)$. Consequently, we add the certificate rule $\sigma(c_1) \rightarrow \vec l_{c_2}-\vec l_n$ to $\RhoL(\CC^\ddagger)$. The variables $\lP^{\CC^\ddagger}, \lQ^{\CC^\ddagger}$ are fixed in $\EE_{\mrm{Config}}$. As a result, $\size{\CC^\ddagger} \leq \exp(\size\CC)$.
        \item \emph{The disjoint case.} Otherwise, $\CC^\ddagger$ lies in $\RhoL(\CC)$ or $\RhoR(\CC)$ and the construction is straightforward. We have $\size{\CC^\ddagger} = O(\size{\CC})$. 
    \end{enumerate}
    The rank analysis is similar to the algebraic decomposition. For the overlapping case, we have $\iota(\CC^\ddagger)\leq \iota(\CC)$ and $\gdim{\CC^\ddagger} < \gdim{\CC}$. For the disjoint case, $\iota(\CC^\ddagger)\leq \iota(\CC)-1$. We can then derive that $\rank{\NKT} < \rank{\KT}$ and $ \size{\NKT} \leq \exp(\size{\KT})$.
\end{proof}

\section{Perfect KLM Trees}
\label{section:perfect_KLM_tree}
In this section, given the reachability proble $\reach(\GG,s,t)$, we denote the size of input by $M \defeq (\size{\GG} + |s| + |t|)$. Binary encoding or unary encoding is irrelevant. We define the following operators with the composition of four basic refinement operators.
\begin{equation*}
    \begin{gathered}
            \clean \defeq \ortho\circ \constr, \\ \refA \defeq \clean\circ\decompA, \\
            \refC \defeq \clean\circ\decompC.
    \end{gathered}
\end{equation*}
For any $\KT$, $\clean(\KT)$ is another KLM tree that is fully-constrained and fully-orthogonalized. If $\KT$ is not unbounded on production rules, we apply $\refA(\KT)$; if $\KT$ is not $\exp$-pumpable, we apply $\refC(\KT)$. In either case, the resulting KLM tree has a smaller rank. Applying the refinement sequence in Figure~\ref{fig:refinement_seq} gives a perfect KLM tree.

\begin{restatable}{theorem}{TheoremSmallKT}
    Given an 1-TGVAS $\GG$ with index $k$ and the reachability problem $\reach(\GG,s,t)$. If there exists a complete, nonnegative derivation tree $\TT$, then there exists a perfect KLM tree $\KT$ such that:
    \begin{enumerate}
        \item $\KT \capture \TT$,
        \item $\size{\KT} \leq g(M)$ for some $g\in \FFF_{2k-1}$.
    \end{enumerate}
    \label{theorem:existence of small perfect KT}
\end{restatable}

\begin{proof}

By Lemma~\ref{lemma:initial_capturing_tree}, there exists an initial capturing KLM tree $\KT$ such that $\KT\capture \TT$ and $\size{\KT}\leq \exp(M)$. Applying $\{\refA,\refC\}^*$ on $\clean(\KT)$ gives a sequence of KLM trees: $\{\KT_0,\KT_1,\dots\}$. We prove the length bound of this sequence by constructing a $(f,m)$-controlled bad sequence:
    \begin{enumerate}
        \item Let $f\in \FFF_2$ be a proper function that bounds the size increase in $\refA$ and $\refC$.
        \item Let $m = f(\exp(M)).$
        \item We define $x_i \defeq \rank{\KT_i}$ for $i\geq 1$. We have $x_i\in \bbn^{2k-2}$. It holds that $\norm{x_i}_{\infty} \leq f^i(m)$ since we have $\norm{\rank{\KT_i}}_{\infty} \leq \size{\KT_i}$. By Lemma \ref{lemma:algebraic decomposition} and Lemma \ref{lemma:combinatorial decomposition}, the sequence $\{x_i\}$ is bad and contains no increasing pair.
    \end{enumerate}
It follows from Lemma~\ref{lemma:bad_seq} that the length function of this sequence (as a function of $m$) is in $\FFF_{2k-2+2-1} = \FFF_{2k-1}$. The size of the last element (the perfect KLM tree) can be represented as $g(M)$ for some $g\in \FFF_{2k-1}$. 
\end{proof}

The following fact is important: a perfect KLM tree $\KT$ is a certificate to the reachability problem itself. 

\begin{restatable}{lemma}{LemmaKTasCertificate}
    Let $\KT \capture \TT$ be a perfect KLM tree. If $\EE_{\KT}$ has a solution, then one can construct, without knowing $\TT$ and merely from $\KT$, another complete, nonnegative derivation tree $\NTT$ such that $\KT \capture \NTT$.
    \label{lemma:KT as certificate}
\end{restatable}

\begin{proof}[Sketch of Proof]
    The same technique as the KLM algorithm for VASS. $\EE_{\KT}$ has a minimal solution $\solmin$. Denote the sum of the Hilbert basis by $\solh$. Our goal is to determine a coefficient $b \in \bbn$ such that $\solmin + b\cdot \solh$ depicts a complete, nonegative derivation tree.
    
    We derive the segment non-negativity for each nontrivial $\CC$ independently. Let its forward- and backward-pumping cycles be $\ccf$ and $\ccb$, respectively. The solution $b\cdot \solh$ serves mainly to repeat $\ccf$ and $\ccb$. Suppose $\CC$ is not orthogonal. For any configuration variable, say $\lPC$, either $\lPC \in \Pump$ or $\lPC$ is unbounded. In both cases, choosing $b$ large enough allows $\lPC$ to be lifted to arbitrarily large values. This guarantees the non-negativity of the entire segment. We refer to Appendix~\ref{section:appendix_perfect_KLM_tree} for the detailed choice of $b$ and the arrangement of the segment.

\end{proof}

The perfectness properties are effectively decidable, according to Lemma~\ref{lemma:decidability_of_orthogonal} and Lemma~\ref{lemma:pumpability is decidable}. Now we are ready to prove Theorem~\ref{theorem:2k}. Given the reachability problem $\reach(\GG,s,t)$, we solve it nondeterministically, by guessing a perfect KLM tree $\KT$ with $\size\KT \leq g(M)$ for $g\in\FFF_{2k-1}$. See Algorithm~\ref{alg:guessKT}.  The correctness relies on the following two facts:

\begin{enumerate}
    \item If the answer to $\reach(\GG,s,t)$ is ``no'', then no guessed KLM tree can be perfect. Otherwise, by Lemma~\ref{lemma:KT as certificate}, a certificate derivation tree $\NTT$ could be constructed, contradicting the answer ``no''.
    \item If the answer to $\reach(\GG,s,t)$ is ``yes'', then, by Theorem~\ref{theorem:existence of small perfect KT}, there exists a small, perfect KLM tree $\KT$, which our nondeterministic algorithm can potentially guess.
\end{enumerate}

\SetKwInOut{Input}{Input}
\SetKwInOut{Output}{Output}

\begin{algorithm}[h]
\caption{Nondeterministic algorithm for 1-TGVAS reachability}
\label{alg:guessKT}

\Input{A $k$-indexed 1-TGVAS $\GG$}
\Output{Whether $\reach(\GG,s,t)$ holds}

\tcp{Guess a perfect $\KT$ with $\size\KT\leq g(M)$.}

\textbf{Step 1:} Guess the topological structure of $\KT$\;
\textbf{Step 2:} For each nontrivial $\CC$, guess its orthogonality\;
\If{not orthogonal}{
    Guess rules and symbols according to $\GG$\;
}
\Else{
    Guess nonterminals, hard-encoded configurations and production rules in $\RhoSCC$\;
    Guess certificate rules in $\RhoL$ or $\RhoR$\;
    Guess the remaining rules and symbols\;
}

\textbf{Step 3:} Guess all terminals and trivial KLM components\;
\textbf{Step 4:} Solve $\EE_{\KT}$; if no solution, \Return{false}\;
\textbf{Step 5:} Check perfectness; if not perfect, \Return{false}\;
\textbf{Step 6:} Recursively check certificate rules with Algorithm~\ref{alg:guessKT}\;
\If{any subproblem rejects}{
    \Return{false}\;
}
\Return{true}.

\end{algorithm}

\TheoremTwok*
\begin{proof}
    Algorithm~\ref{alg:guessKT} guesses a perfect $\KT$ with $\size{\KT}\leq g(M)$ for $g\in \FFF_{2k-1}$. The check for certificate rules is carried out by invoking the reachability subproblem. There are at most $g(M)$ subproblems like $\reach(\GG_A,s_A,t_A)$ with $\iota(\GG_A)<k$. The input size is given by $(\size{\GG_A} + |s_A| + |t_A|)$ and is bounded by $g(M)$. We prove that $\reach(\GG,s,t)$ with $\iota(\GG)=k$ is solvable in $g_{k}(M)$ for $g_{k}\in \FFF_{2k-1}$ by the following induction.
\begin{itemize}
\item [$k=1$.]
The GVAS is trivial (in the form of $S_0\rightarrow \vec u$). The reachability problem is solvable in $\FFF_1$.
\item [$k>1$.] We guess a perfect KLM tree in $g(M)$ for $g\in\FFF_{2k-1}$. We check at most $g(M)$ subproblems whose sizes are bounded by $g(M)$ and indices no greater than $k-1$. By assumption, these subproblems can be solved in:
\[g(M)\cdot g_{k-1}(g(M))\]
steps, where $g_{k-1}$ is a function in $\FFF_{2(k-1)-1} = \FFF_{2k-3}$. This function is in $\FFF_{2k-1}$.
\end{itemize} 
In conclusion, the reachability for $k$-indexed 1-TGVAS is in $\FF_{2k}$.
\end{proof}

\section{Conclusion}
\label{section:conclusion}
We have shown that the reachability problem for $k$-indexed 1-TGVAS is in $\FF_{2k}$, improving the previous known $\FF_{6k-4}$. We generalized the KLM decomposition method to KLM trees, providing an effective characterization of the $\bbz$-reachability of TGVAS.

\emph{Future research.} The exact complexity of the reachability problem for 1-GVAS remains open. A \textbf{PSPACE}-hardness is obtained in~\cite{PSPACElowerbound}. An Ackermannian upper bound follows from the reduction from general 1-GVAS to 1-TGVAS~\cite{DBLP:conf/stoc/BiziereC25_BC25}. Achieving a non-Ackermannian upper bound would require non-KLM-based approaches, such as analyzing the smallest derivation tree. For $d$-TGVAS, there is still room for improvement in the existing $\FF_{4kd+2k-4d}$ upper bound. The approximation method in~\cite{DBLP:conf/lics/GuttenbergCL25_CmeVASS} is widely applicable. By combining this approximation technique with the KLM tree, we conjecture an upper bound of $\FF_{2kd}$. Finally, the validity and necessity of the index as a measure remain open questions; for instance, whether there exists a lower bound parameterized by the index is to be investigated.

\bibliographystyle{ACM-Reference-Format}
\bibliography{ref}

\appendix
\section{Proofs in Section \ref{section:GVAS}}
\label{section:appendix_GVAS}
In this section, first we show the equivalence of different definitions of thinness and index. Then, we show the complexity upper bound for $k$-indexed $d$-TGVAS obtained by \cite{DBLP:conf/lics/GuttenbergCL25_CmeVASS}.

\subsection{Thinness and Index}
\DefinitionThinness*

We show that the definition in the main text is equivalent to what we state in the following lemma:

\begin{lemma}
    A GVAS $\GG$ is thin, iff for any $X\in\Chi$, $X$ cannot derive more than one copy of itself.
\end{lemma}

\begin{proof}
    $(\Rightarrow)$ $\GG$ is thin. Suppose that $X$ has the derivation $X\Rightarrow \alpha X\beta X \gamma$. The corresponding derivation tree is $\TT_X$. Let the lowest common ancestor of the two $X$'s be node $p$. $p$ has the rule $\sigma(p)\rightarrow \sigma(c_1)\sigma(c_2)$. Since $\sigma(c_1)$ derives $X$ and $\sigma(c_2)$ derives $X$, $\sigma(c_1),\sigma(c_2),X$ are in the same SCC. The rule $\sigma(p)\rightarrow \sigma(c_1)\sigma(c_2)$ is nondegenerate. Contradiction. We conclude that $X$ cannot derive more than one copy of itself.
    $(\Leftarrow)$ $X$ cannot derive more than one copy of itself. Suppose that there exists a nondegenerate rule of $X$: the rule $X\rightarrow YZ$. Since $X,Y,Z$ are in the same SCC, it is immediate that $Y$ derives $X$ and $Z$ derives $X$. $X$ can derive two $X$'s. Contradiction. There exists no nondegenerate rule.
\end{proof}

Bizi{\`{e}}re et al. point out that ``every thin GVAS has a finite index'' in \cite{DBLP:conf/stoc/BiziereC25_BC25}, Claim 5. We show that every GVAS with a finite index is thin. Here the index is given by: \[\iota(\GG) = \sup\left\{ k_{S_0\Rightarrow \alpha} \mid S_0\Rightarrow \alpha \text{ is optimal} \right\}.\] Given a GVAS $\GG$, $\GG$ has a finite index. We can infer that for any $X\in\Chi$, the index of $\GG_X$ is finite. Assume that $\GG$ is not thin. There exists a nondegenerate rule $X\rightarrow YZ$ where $X,Y,Z$ are in the same SCC of $\PG$.  We have $\iota(\GG_Y) \geq \iota(\GG_X)$ and $\iota(\GG_Z)\geq \iota(\GG_X)$. Then the fact that $\iota(\GG_X)\geq \iota(\GG_Y)+1,\iota(\GG_X)\geq \iota(\GG_Z)+1$ contradicts the finite index. We conclude that every finitely indexed GVAS is thin.

As for thin GVAS, we show the equivalence of different definitions of the index.
\DefinitionIndex*

\begin{lemma}
Given thin GVAS $\GG$, denote a complete derivation as $S_0\Rightarrow \alpha$ where $\alpha \in \Sigma^*$. Then $\iota(\GG)$ defined in \ref{definition:index} satisfies:
\[\iota(\GG) = \sup\left\{ k_{S_0 \Rightarrow \alpha} \mid S_0\Rightarrow \alpha \text{ is optimal} \right\}.\]
\end{lemma}

\begin{proof}
According to our dynamic programming process, $\iota(\GG)$ is an upper bound of nonterminal appearances in optimal derivations. To prove that $\iota(\GG)$ is tight, one can easily induct the tightness from lower SCCs to the top.
\end{proof}

\subsection{Upper Bounds for $k$-Indexed $d$-TGVAS}

Basically, a $d$-VASS is a directed graph whose edges are labeled by vectors in $\bbz^d$. The semantics of an edge is as follows: for an edge $(p,q)$ labeled by $\vec w$, a configuration $p(\vec u)$ can transfer to a configuration $q(\vec v)$ if $\vec u,\vec v\in \bbn^d$ and $\vec u+\vec w = \vec v$.

Guttenberg, Czerwi\'nski and Lasota introduced an extended VASS model in \cite{DBLP:conf/lics/GuttenbergCL25_CmeVASS}. Let $\CC$ be a class of relations of $\bbn^d \rightarrow \bbn^d$. Notice that the relation class $\CC$ here should not be confused with the KLM component $\CC$ introduced later. Then, a $\CC$-extended VASS is a directed graph whose edges are labeled by relations in $\CC$. The semantics of an edge is: for $(p,q)$ labeled by relation $r:\bbn^d\rightarrow \bbn^d, r\in \CC$, a configuration $p(\vec u)$ can transfer to $q(\vec v)$ if $\vec u,\vec v \geq \vec 0$ and $r(\vec u) = \vec v$.

Functions, relations and sets of vectors are equivalent. A function or relation of $\bbn^d\rightarrow \bbn^d$ is equivalent a set on $\bbn^d \times \bbn^d$. A relation $r$ is \emph{monotone}, if for $\vec u,\vec v \in \bbn^d$, as long as $r(\vec u)$ is defined, we have $r(\vec u + \vec v) = r(\vec u) + \vec v$. In the set perspective, a set on $\bbn^d\times \bbn^d$ is monotone, if for $(\vec x, \vec y)\in S$ and any $\vec z \in \bbn^d$, it holds that $(\vec x + \vec z, \vec y + \vec z)\in S$.  A $\CC$-extended VASS is monotone, if every edge relation in $\CC$ is additionally monotone. Monotone $\CC$-extended VASS are abbreviated as $\CC$-m-eVASS.

Guttenberg, Czerwi\'nski and Lasota proposed a universal algorithm to approximate the reachability set of $\CC$-m-eVASS. The concept of ``approximation'' is out of our scope. It is a fact that: if a reachability set can be approximated, the corresponding reachability problem is easily decided. Their algorithm highly depends on the following theorem:

\begin{theorem}[Theorem V.5. in \cite{DBLP:conf/lics/GuttenbergCL25_CmeVASS}]
    Let $\alpha\geq 2$ and let $\CC$ be a class of relations approximable in $\FF_{\alpha}$, containing $\mrm{Add}$ and effectively closed under intersection with semilinear relations in $\FF_{\alpha}$. Then sections of $\CC$-m-eVASS are approximable in $\FF_{\alpha+2d+2}$.
    \label{Theorem:Guttenberg}
\end{theorem}

We explain several notions in the theorem. Firstly, the fast-growing hierarchy defined in \cite{DBLP:conf/lics/GuttenbergCL25_CmeVASS} coincides with our definition of $\FF_d$. Secondly, the ``section'' of $\CC$-m-eVASS stands for the intersection of the reachability relation of $\CC$-m-eVASS and a semilinear set. Moreover, the projection of a section is also called a section. Thirdly, the monotone section of $\CC$-m-eVASS contains the relation $\mrm{Add}$ and is effectively closed under intersection with semilinear sets indeed.

We aim to reduce a $k$-indexed $d$-TGVAS to a $\CC$-m-eVASS by Lemma~\ref{lemma:CmeVASSreduction}. Then, we obtain the complexity $\FF_{4kd+2k-4d}$ upper bound by induction in Lemma~\ref{Lemma:4kd}.

Let the identical function be $\mrm{Id}:\bbn^d \rightarrow \bbn^d$ and $\mrm{Id}(\vec x)\defeq \vec x$. Let the add function be $\mrm{Add}_{\vec u} : \bbn^d \rightarrow \bbn^d$ and $\mrm{Add}_{\vec u}(\vec x) \defeq \vec x + \vec u$. We sometimes concatenate functions, for instance $(\mrm{Id},\mrm{Id})$, which means then identical function of $\bbn^{2d}\rightarrow \bbn^{2d}$. Moreover, the concatenation $(\mrm{Id}, \mrm{Add}_{\vec u})$ stands for a function of $\bbn^{2d}\rightarrow \bbn^{2d}$, which copies the first $d$ digits and add $\vec u$ to the second $d$ digits. The concatenation of monotone functions is still monotone.

Additionally, we define the mirror of TGVAS $\GG(\Chi, \Sigma, \Rho, S_0)$ as $\GG^{\mrm{Mir}} \defeq (\Chi, \Sigma^{\mrm{Mir}}, \Rho^{\mrm{Mir}}, S_0)$ where:
\begin{enumerate}
    \item For $\vec u\in \Sigma$, we have $(-\vec u) \in \Sigma^{\mrm{Mir}}.$
    \item For $X\rightarrow AB \in \Rho$, we have $X\rightarrow BA \in \Rho^{\mrm{Mir}}$. 
    \item For $X\rightarrow \vec u\in \Rho$, we have $X\rightarrow (-\vec u)\in \Rho^{\mrm{Mir}}$.
\end{enumerate}

\begin{lemma}
    The reachability relation of a $k$-indexed $d$-TGVAS ($k\geq 2$) is equivalent to a finite union of reachability section projections of $2d$-dimensional $\CC$-m-eVASSes, where $\CC$ is the class of reachability sections of $(k-1)$-indexed $d$-TGVAS.
    \label{lemma:CmeVASSreduction}
\end{lemma}

\begin{proof}
    Given a $k$-indexed $d$-TGVAS $\GG$. We define the set of nonterminals:
    \[\Chi_{\iota=k} \defeq \left\{ X\mid X\in \Chi, \iota(\GG_X) = k \right\}. \]
    Consider the reachability problem $\reach(\GG,\vec s, \vec t)$. Suppose that there exists a nonnegative, complete derivation tree $\TT$, whose root is node $r$. Nodes, say $p$, with $\sigma(p) \in \Chi_{\iota = k}$, lie in a path in $\TT$. The path is $\ppath{r,q}$. We construct a $2d$-dimensional $\CC$-m-eVASS $\mathcal{V}$ to depict the configurations in $\ppath{r,q}$. Let the states of $\mathcal{V}$ be $\Chi_{\iota=k}$. The construction is straightforward: 

    \begin{enumerate}
        \item For any $p\in \ppath{r,q}$ and $p\neq q$, if $p$ has the production rule $\sigma(p)\rightarrow \sigma(c_1)\sigma(c_2)$ and the rule is left-degenerate, add edge $(\sigma(p),\sigma(c_2))$ to the edge set of $\mathcal{V}$. The label is determined by the rules used in $\sub{c_1}$. Denote the sub-GVAS induced by them as $\GG_{\sub{c_1}}$. Denote the reachability relation of $\GG_{\sub{c_1}}$ as $R(\GG_{\sub{c_1}})$. The reachability relation is in $\bbn^d \times \bbn^d$. The label for the edge $(\sigma(p),\sigma(c_2))$ is:
        \[(R(\GG_{\sub{c_1}}), \mrm{Id}) \in \bbn^{2d}\times \bbn^{2d}.\]
        \item For $p\in \ppath{r,q}$ and $p\neq q$, if $\sigma(p)\rightarrow \sigma(c_1)\sigma(c_2)$ is right-degenerate, we add the edge $(\sigma(p),\sigma(c_1))$ labeled by:
        \[(\mrm{Id}, R(\GG^{\mrm{Mir}}_{\sub{c_2}})) \in \bbn^{2d}\times \bbn^{2d}.\]
        \item For the node $q$, if $q$ has the rule $\sigma(q)\rightarrow \vec u$ in $\TT$, we create a new state $\mrm{Acc}$ and obtain the edge $(\sigma(q),\mrm{Acc})$. The edge label is:
        \[(\mrm{Add}_{\vec u},\mrm{Id}) \in \bbn^{2d}\times \bbn^{2d}.\]
        \item If $q$ has the rule $\sigma(q)\rightarrow \sigma(c_1)\sigma(c_2)$, we create the state $\mrm{Acc}$ and add the edge $(\sigma(q), \mrm{Acc})$ labeled by:
        \[(R(\GG_{\sub{c_1}}), R(\GG^{\mrm{Mir}}_{\sub{c_2}})) \in  \bbn^{2d}\times \bbn^{2d}.\]
    \end{enumerate}

    The edge relations are monotone. The derivation tree $\TT$ is mapped to a nonnegative run of $\mathcal{V}$ from $S_0$ to $\mrm{Acc}$. The configuration of $S_0$ is $(\vec s,\vec t)$. The configuration of $\mrm{Acc}$ is in the form of $(\vec a,\vec a)$ where $\vec a\in \bbn^d$. Together, a successful run in the reachability relation of $\mathcal{V}$ is the vector $(\vec s,\vec t, \vec a,\vec a)$. Let the semilinear set $S \subseteq \bbn^{4d}$ be:
    \[S\defeq \left\{ (\vec b,\vec c,\vec a,\vec a) \mid \vec a,\vec b,\vec c\in \bbn^d \right\}.\]
    It holds that $S$ is semilinear and monotone. Then, the relation $(\vec s,\vec t)$ is obtained as follows. First, intersect the reachability relation of $\mathcal{V}$ with $S$. Then, project the section into the first $2d$ digits. We conclude that the reachability relation of $\GG$ is a finite union of reachability section projections of $2d$-dimensional $\CC$-m-eVASSes like $\mathcal{V}$. By the definition of partially-degenerate rules and the index, $\CC$ contains only the reachability sections of $(k-1)$-indexed $d$-TGVAS.    
\end{proof}

\begin{lemma}
    The reachability problem for $k$-indexed $d$-TGVAS is in $\FF_{4kd+2k-4d}$.
    \label{Lemma:4kd}
\end{lemma}

\begin{proof}
    By the approximation technique introduced in \cite{DBLP:conf/lics/GuttenbergCL25_CmeVASS}, approximation indicates decidability. It is sufficient to prove that the reachability sections of $k$-indexed $d$-TGVAS are approximable in $\FF_{4kd+2k-4d}$. We give the proof by induction.

    \begin{itemize}
        \item [$k=1.$] The TGVAS is trivial and the reachability sections can be approximated in $\FF_2$.
        \item [$k>1.$] Assume that the condition holds for $(k-1)$. Let the class of reachability sections of $(k-1)$-indexed $d$-TGVAS be $\CC$. By Theorem~\ref{Theorem:Guttenberg}, the reachability sections of $\CC$-m-eVASS are approximable in:
        \[\FF_{4(k-1)d + 2(k-1) -4d + 4d +2} = \FF_{4kd+2k-4d}.\]
        Then, according to Lemma~\ref{lemma:CmeVASSreduction}, the reachability sections of $k$-indexed $d$-TGVAS are effectively represented by the reachability sections of $\CC$-m-eVASS. We conclude that the reachability sections of $k$-indexed $d$-TGVAS are approximable in $\FF_{4kd+2k-4d}$.
    \end{itemize}
\end{proof}

The approximation method proposed in \cite{DBLP:conf/lics/GuttenbergCL25_CmeVASS} is significant. For 1-TGVAS, it provides a decent approach (an upper bound of $\FF_{6k-4}$) to the brute-force barrier (known as $\FF_{k}$, or Ackermannian). The brute-force barrier is the size of a possible, finite reachability set. We show that the finite reachability set of $k$-indexed 1-TGVAS is at least $\FF_{k}$ by constructing a $k$-indexed 1-TGVAS that weakly computes the fast-growing function $F_{k}$.

For convenience, we allow a slight violation of the Chomsky normal form. The nonterminal set is $\Chi = \{X_1,X_2,\dots ,X_{k}\}$ where $X_i$ weakly computes $F_i$. For $X_1$, we have the rules: $X_1 \rightarrow -1 X_12 \in$ and $X_1 \rightarrow 0$. For $X_i$ with $i\geq 2$, we have the rules: $X_i\rightarrow -1 X_iX_{i-1}$ and $X_i \rightarrow 1$. Such 1-TGVAS has a finite reachability set with an Ackermannian size.

\section{Proofs in Section \ref{section:decomposing_KLM_tree}}
\label{section:appendix_decomposing_KLM_tree}
In this section, we complete the proof of refinement lemmas in section \ref{section:decomposing_KLM_tree}.

\subsection{Orthogonalization}
We show the decidability of orthogonality and explain the computation of the uniform displacement $\vec u_A$ for $A\in \ChiL$.
\LemmaDecidabilityOfOrthogonal*

\begin{proof}
    Suppose we have $\CC\capture \seg{p,q}$ that is left-orthogonal. The following conditions are jointly equivalent to left-orthogonality:
    \begin{enumerate}
        \item Given $A\in \ChiL$ and the induced sub-GVAS of $\GG_A$ under the rule of $\RhoL$. For any cycle in $\GG_A$, its displacement is $0$.
        \item For such $A$, any complete derivation tree $\TT_A$ has fixed displacement, known as $\vec u_A$. 
        \item Consider the weighted, directed graph $G(\ChiSCC, E)$ constructed by mapping left-degenerate rules like $X\rightarrow AY \in \RhoSCC$ to the edge $(X,Y,\vec u_A)$ and right-degenerate rules like  $X\rightarrow YB\in \RhoSCC$ to the edge $(X,Y,0)$. Any cycle in $G$ has a weight sum of $0$.
    \end{enumerate}

    We show the equivalence. $( \text{Orthogonality} \Rightarrow \text{claims})$ The second claim covers the first one. Let $\TT_A$ and $\TT_A^\ast$ be two complete derivation trees with different displacements. By the definition of $\CC \capture \seg{p,q}$, we can construct a top cycle in $ {\mathcal{O}}_{\CC}$ with $A$ derived on the left at least once. Replace the subtree of $A$ by $\TT_A$ and $\TT_A^\ast$, respectively. The two top cycles obtained have different left effects and cannot both be $0$, contradicting the left-orthogonality. If the first two conditions hold, then any top cycle in $ {\mathcal{O}}_{\CC}$ can be mapped to a cycle in $G$. The third claim is immediate by this mapping. $(\text{Claims}\Rightarrow\text{Orthogonality})$ By these claims, any $A\in \ChiL$ has a fixed displacement. By claim (3)  the left-orthogonality is immediate.

    Next, we explain how to check them. Suppose that the KLM component $\CC$ is input under binary encoding. Claim (1) can be checked by enumerating simple cycles in $\GG_A$ in $\EXP(|\ChiL|)$ space. Claim (2) can be checked by enumerating complete derivation trees without cycles since cycles are useless. It is in $\EXP(|\ChiL|)$ space and the uniform displacement $\vec u_A$ is thereby computed. The uniform displacement has a bound of $\exp(\size\CC)$, which is not tight but enough. Claim (3) can be checked by constructing such $G$ and enumerating simple cycles within $\poly{|\ChiSCC|}$ space. The overall complexity is in $\EXPSPACE$.
\end{proof}

We emphasize that no nontrivial KLM component $\CC$ is both left- and right-orthogonalized. Indeed, a KLM component that is both left- and right-orthogonal has geometric dimension $0$, which contradicts its nontriviality.

\subsection{Algebraic Decomposition}
We prove the size bound of the division of segments.

\LemmaSizeofDivision*
\begin{proof}
    For any bounded derivation variable $\#(\rho)$, we have $\#(\rho) \leq \exp(\size{\KT})$ by Lemma \ref{lemma:pottier}. The number of different variables $\#(\rho)$ in $\CC$ is at most $\size{\CC}$. Therefore, the size $|N_{\Bounded}|$ is bounded by $\exp(\size{\KT})$.

    A virtual tree is the minimal subtree that is closed under the lowest common ancestor operation $\mrm{lca}(\cdot)$. For any $n \in \mrm{LCA}(N_{\Bounded})$, either of the following is true:
    \begin{enumerate}
        \item $n\in N_{\Bounded}$.
        \item $n\notin N_{\Bounded}$ and $n$ has two children in the virtual tree.
    \end{enumerate}
    We can infer that the size $|\mrm{LCA}(N_{\Bounded})|$ is at most $(2\cdot |N_{\Bounded}|-1)$. The worst case happens when all $n\in N_{\Bounded}$ are leaves.

    The segment fracturing step produces at most $\exp(\size{\KT})$ segments of $\seg{p^\ast,q^\ast}$. For each $\seg{p^\ast,q^\ast}$, the SCC division produces at most $\size{\CC}$ strongly-connected segments of $\seg{p^\dagger, q^\dagger}$, since the number of nonterminals is limited by $\size\CC$ and nonterminals do not overlap in different SCCs.

    In the subtree completion step, for any $\seg{p^\dagger, q^\dagger}$, it has at most two uncovered children. Consider an uncovered child $c_1$. We have $c_1\notin \ppath{p,q}$, since $\ppath{p,q}$ is fully fractured (That is, in the segment fracturing step, all nodes in $\ppath{p,q}$ are included in some segment $\seg{p^\ast,q^\ast}$).
    
    Therefore, $\sub{c_1}$ is a complete derivation tree. Rules used in $\sub{c_1}$ is the subset of $\RhoL(\CC)$ or $\RhoR(\CC)$. We construct the induced sub-GVAS $\GG_{\sigma(c_1)}$ and we have $\size{\GG_{\sigma(c_1)}} \leq \size{\CC}$.

    As a result, the number of segments in $\subdiv{\sub{c_1}}$ is limited by $\exp(\size{\CC})$ by Lemma \ref{lemma:small_division}. The overall segments in $\subdiv{\seg{p,q}}$ is bounded by:
    \begin{equation*}
    \begin{aligned}
    |\subdiv{\seg{p,q}}| &\le \exp(\size{\KT}) \cdot \size{\CC} \cdot \exp(\size{\CC}) \\
                          &\le \exp(\size{\KT}) .
    \end{aligned}
    \end{equation*}
\end{proof}

As for the algebraic decomposition, it remains to prove the strict decrease in rank.

\LemmaAlgebraicDecomposition*
\begin{proof}
    In the main text, we constructed new KLM components $\CC^\ddagger$ in order to capture $\seg{p^\ddagger, q^\ddagger}\in\subdiv{\seg{p,q}}$. We showed that $\size{C^\ddagger} = O(\size{\CC})$. The number of $\seg{p^\ddagger, q^\ddagger}$ are bounded by Lemma \ref{lemma:size bound for segment division}. Therefore, we have $\size{\NKT} \leq \exp(\size{\KT})$. 

    In order to prove $\rank{\NKT}< \rank{\KT}$, it suffices to show that for every nontrivial $\CC^\ddagger \capture \seg{p^\ddagger, q^\ddagger}$, $\rank{\CC^\ddagger} < \rank{\CC}$. Newly-constructed $\CC^\ddagger$ are distinguished by whether $\ppath{p^\ddagger, q^\ddagger} \subseteq \ppath{p,q}$ or not.
    
    \emph{The overlapping case.} If $\ppath{p^\ddagger, q^\ddagger} \subseteq \ppath{p,q}$, we can infer that:
    \begin{equation*}\begin{aligned}
        \RhoSCC(\CC^\ddagger) &\subseteq \RhoSCC(\CC), \\
        \RhoL(\CC^\ddagger) &\subseteq \RhoL(\CC), \\
        \RhoR(\CC^\ddagger) &\subseteq \RhoR(\CC).
    \end{aligned}\end{equation*}
    Since removing some production rules does not increase the index and the geometric dimension, it is immediate that $\rank{\CC^\ddagger}\leq \rank{\CC}$. We prove that $\gdim{\CC^\ddagger} < \gdim{\CC}$ by contradiction.

    Assume that $\gdim{\CC^\ddagger} = \gdim{\CC}$. That is, the dimension of the vector space $\bbq(\Delta( {\mathcal{O}}_{\CC^\ddagger}))$ is equal to the dimension of $\bbq(\Delta( {\mathcal{O}}_{\CC}))$, denoted by $d_{\mrm{G}}$. We can infer that $\bbq(\Delta( {\mathcal{O}}_{\CC^\ddagger})) = \bbq(\Delta( {\mathcal{O}}_{\CC}))$. There exists a set of bases $V = \{\vec v_1, \dots ,\vec v_{d_{\mrm{G}}}\}\subseteq \bbn^2$ such that:

    \begin{enumerate}
        \item $V \subseteq \Delta( {\mathcal{O}}_{\CC^\ddagger})$. Every $\vec v_i \in V$ is the effect of a top cycle in $ {\mathcal{O}}_{\CC^\ddagger}$.
        \item $\bbq(V) = \bbq(\Delta( {\mathcal{O}}_{\CC^\ddagger})) = \bbq(\Delta( {\mathcal{O}}_{\CC})).$
    \end{enumerate}

    Since $\CC\capture \seg{p,q}$, we can construct a top cycle $\cyc{x,y}$ containing all production rules in $\RhoSCC(\CC), \RhoL(\CC), \RhoR(\CC)$. We have $\Delta(\cyc{x,y}) \in \bbq(V)$. Consider the representation of the linear combination of $\vec v_i\in V$.
    \[\Delta(\cyc{x,y}) = \sum_{i=1}^{d_\mrm{G}} \lambda_i \vec v_i, \; \lambda_i \in \bbq.\]
    We obtain an equation with coefficients in $\mathbb{N}_+$ by removing zero coefficients, moving the terms with negative $\lambda_i$ to the left, and multiplying by a common denominator.
    \[\mu\Delta(\cyc{x,y}) +\sum_{\lambda_i<0} \mu_i\vec v_i  = \sum_{\lambda_i >0} \mu_i \vec v_i, \; \mu,\mu_i \in \bbn_+. \tag{*}\label{eq:star}\]
    By the segment fracturing step, $\CC^\ddagger$ contains no bounded rules. The left-hand side of (\ref{eq:star}) is the total effect of some cycles, where $\cyc{x,y}$ contains at least one bounded rule. The right-hand side of (\ref{eq:star}) is the total effect of some cycles in $\CC^\ddagger$ (or $0$). Equation (\ref{eq:star}) indicates that bounded rules in $\CC$ are unbounded. Contradiction. We have $\gdim{\CC^\ddagger} < \gdim{\CC}$.

    \emph{The disjoint case.} If $\ppath{p^\ddagger, q^\ddagger} \subseteq \ppath{p,q}$ does not hold, by the property of the virtual tree we can infer that $\ppath{p^\ddagger, q^\ddagger} \cap \ppath{p,q} = \varnothing$. The new component $\CC^\ddagger$ lies in the rule set $\RhoL(\CC)$ or $\RhoR(\CC)$. For any nonterminal $X^\ddagger$ occurring in $\ppath{p^\ddagger, q^\ddagger}$, we have $\iota(\GG_{X^\ddagger}) \leq \iota(\CC)$ since either $X^\ddagger\in \ChiL(\CC)$ or $X^\ddagger\in \ChiR(\CC)$. By the definition of index, for any $A^\ddagger \in \ChiL(\CC^\ddagger)$ or $A^\ddagger \in \ChiR(\CC^\ddagger)$, we have $\iota(\GG_{A^\ddagger}) \leq \iota(\GG_{X^\ddagger}) -1$. The index of $\CC^\ddagger$ is the largest index of such $\GG_{A^\ddagger}$, which means $\iota(\CC^\ddagger) \leq \iota(\CC) -1$.

    By the two cases above, we conclude that $\rank{\CC^\ddagger} < \rank{\CC}$ for nontrivial $\CC^\ddagger$ and $\rank{\NKT} < \rank{\KT}$.
\end{proof}

We did not show the algebraic decomposition for the orthogonalized cases explicitly, since we just treat nonterminals with hard-encoded configurations like $(_3X)$ or $(Y_5)$ as normal symbols. It is unnecessary to add equations to limit their configurations, since they are always fixed and are constrained in the next configuration constraining step.

\subsection{Combinatorial Decomposition}

Here, we discuss pumpability. We complete proofs in the main text, considering only the forward-pumpability.

Consider $\CC \capture \seg{p,q}$. In this subsection, we require $\CC$ be fully constrained. That is, the bounded values $\lPC$ or $\rPC$ are included in $\size\CC$. Given $A\in \ChiL$ and $(s, t) \in \bbn\times \bbn$. The coverability problem $\cover(\GG_A, s, t)$ is decidable in {\EXPSPACE} \cite{DBLP:conf/icalp/LerouxST15_Coverability}. We introduce the coverability upper bound function. Recall the mirror of TGVAS introduced in Appendix~\ref{section:appendix_GVAS}.

\begin{definition}
    Given $A\in \ChiL$ and let $\GG_A$ be the induced sub-GVAS under $\RhoL$. The forward-coverability upper bound function $\delta_A(s) : \bbnww \rightarrow \bbnww$ is defined as:
    \begin{enumerate}
        \item $\delta_A(-\omega) \defeq -\omega, \delta_A(+\omega) \defeq +\omega$.
        \item If $\cover(\GG_A, s, 0)$ rejects, then $\delta_A(s) \defeq -\omega$.
        \item Otherwise, $\delta_A(s) \defeq \sup\{t\in\bbn \mid \cover(\GG_A, s, t) \text{ accepts}\}$.
    \end{enumerate}
\end{definition}
\begin{definition}
    For $B\in \ChiR$, the function $\delta_B(s) : \bbnww \rightarrow \bbnww$ is defined according to $\GG_B^{\mrm{Mir}}$:
    \begin{enumerate}
        \item $\delta_A(-\omega) \defeq -\omega, \delta_A(+\omega) \defeq +\omega$.
        \item If $\cover(\GG_B^{\mrm{Mir}}, s, 0)$ rejects, then $\delta_B(s) \defeq -\omega$.
        \item Otherwise, $\delta_B(s) \defeq \sup\{t\in\bbn \mid \cover(\GG_B^{\mrm{Mir}}, s, t) \text{ accepts}\}$.
    \end{enumerate}
\end{definition}

The forward-coverability upper bound function is monotone. For example, given $\delta_A(\cdot)$ and $s,t \in \bbnww$, if $s\leq t$ then $\delta_A(s)\leq \delta_A(t)$. The \emph{truncated} version of function $\delta_A(\cdot)$ (or $\delta_B(\cdot)$) are effectively computable. Given $\GG_A, s\in \bbn$ and the bound value $b\in \bbn$, the truncated function: \[\min\{\delta_A(s), b\}\] is computable in \EXPSPACE\  by enumerating $t \in [0,b]$ and ask the coverability problem $\cover(\GG_A, s, t)$.

The following lemma provides a lower bound for the coverability upper bound function:

\begin{lemma}
    There exists some $\Delta \leq \exp(\size{\CC})$ such that for every acyclic, complete derivation tree under $\RhoL$, its displacement is no less than $(-\Delta)$. For $s\in \bbn, s\geq \Delta$, we have $\delta_A(s) \geq s - \Delta$. The same $\Delta$ applies for $\RhoR^{\mrm{Mir}}$ and $\delta_B(\cdot)$.
    \label{Lemma:Lower bound for delta} 
\end{lemma}

\begin{proof}
    Given $A \in \ChiL$, any acyclic, complete derivation tree under the proper induced sub-GVAS $\GG_A$ has height bound $\size{\CC}$. Therefore, it has at most $\exp(\size{\CC})$ leaves. Its displacement is no less than $-\exp(\size{\CC})$. Such acyclic, complete derivation tree is called the \emph{fastest degenerate tree}. Take $\Delta$ as the biggest absolute value of the displacements of these fastest degenerate trees.
\end{proof}

Intuitively, in the forward-pumpability, when we  construct the pumping cycle $\cyc{x,y}$, a left-degenerate rule $X\rightarrow AY\in \RhoSCC$ produces subtree rooted by $A$, whose displacement should be maximized. Therefore, we can replace $A$ with the upper bound function $\delta_A(\cdot)$.

The \emph{forward-coverability relation extension graph} of KLM component $\CC\capture \seg{p,q}$ is an edge-labeled, directed graph $G(\ChiSCC, E)$. The label on the edge is a relation on $\bbn^2\times \bbn^2$. Let the identical function be $\mrm{Id}:\bbn \rightarrow \bbn$ and $\mrm{Id}(x)\defeq x$. The graph $G$ is constructed by:
\begin{enumerate}
    \item For a left-degenerate rule $X\rightarrow AY\in \RhoSCC$, add the edge $(X,Y)$ labeled by $(\delta_A(\cdot), \mrm{Id}(\cdot))$ to $E$.
    \item For a right-degenerate rule $X\rightarrow YB\in \RhoSCC$, add the edge $(X,Y)$ labeled by $(\mrm{Id}(\cdot), \delta_B(\cdot))$ to $E$.
\end{enumerate}

Given $(v_1,v_2)\in E$, we denote the relation by $\mrm{Rel}_{(v_1,v_2)}$. Actually, $G$ is a monotone, relation-extended-2-VASS defined in \cite{DBLP:conf/lics/GuttenbergCL25_CmeVASS}. The forward-pumpability of $\CC$ is equivalent to the forward-coverability on $G$. We explain this on the one-sided pumpability case.

Suppose $\lPC\in \Pump, \rPC\notin \Pump$ and there exists a pumping cycle $\cyc{x,y}$ with $\vec l_x = \lPC$ and $\vec l_y\geq \vec l_x+1$. Let the length $|\ppath{x,y}|$ be $h$. Then, there exists a \emph{pumping sequence} in $G$, denoted by $\{v_1, v_2,\dots v_{h}\}$, satisfying:
\begin{enumerate}
    \item $v_1 = v_h = P$, representing that the pumping cycle $\cyc{x,y}$ satisfies $\sigma(x) = \sigma(y)=P$.
    \item Every $v_i$ is labeled with a configuration $(\vec l_{v_i},\vec r_{v_i}) \in \bbnww \times \bbnww$. It holds that $(\vec l_{v_1},\vec r_{v_1}) = (\vec l_x,\vec r_x) = (\lPC, +\omega)$ and $(\vec l_{v_h},\vec r_{v_h}) \geq (\vec l_y, \vec r_y)$.
    \item For $v_i$ and its successor $v_{i+1}$, their configurations satisfy $(\vec l_{v_{i+1}},\vec r_{v_{i+1}}) = \mrm{Rel}_{(v_{i}, v_{i+1})} (\vec l_{v_{i}},\vec r_{v_{i}})$.
\end{enumerate}
The existence is immediate from the fact that $\delta_A(\cdot),\delta_B(\cdot)$ are the supremum of displacements of derivation trees.

On the other hand, if such $\{v_1,\dots, v_h\}$ exists, we can construct a pumping cycle $\cyc{x,y}$ with the corresponding coverability certificates. Notice that for the infinite case, i.e. for some $i\in[1,h]$ such that $\vec l_{v_i} < +\omega$ and $\vec l_{v_{i+1}} = +\omega$, we only need a complete derivation tree with displacement no less than $(\lPC + h\cdot \Delta + 1)$ to guarantee the pumpability by Lemma \ref{Lemma:Lower bound for delta}.

\emph{The fastest returning path.} Given $v\in G$, a configuration $(\vec l_v,\vec r_v)$ and a target $u\in G$. If $(\vec l_v,\vec r_v) \geq ( a + \Delta\cdot |\ChiSCC|, b + \Delta\cdot |\ChiSCC|)$ for $a,b\in \bbnw$, then there exists a \emph{fastest returning path} $\{v_1,v_2,\dots v_m\}$ with:
\begin{enumerate}
    \item $v_i\in G, v_1 = v, v_m = u$.
    \item The path is acyclic, i.e. $m\leq |\ChiSCC|$.
    \item The configuration $(\vec l_u, \vec r_u) \geq (a,b)$.
\end{enumerate}

The existence of a fastest returning path is immediate by Lemma \ref{Lemma:Lower bound for delta}.

\LemmaOneSidedPumping*
\begin{proof}
    First, we distinguish different bounds which are all written as $\exp(\size\CC)$ in the statement. We refer to $\vec l_n\geq b_1$ and $|\ppath{x,y}|\leq b_2$. We show the existence of these bounds.
    
    Consider the subpath $\ppath{p,n}$ with $|\ppath{p,n}|=h$. We construct the sequence $\{v_1,v_2,\dots, v_h\}$ in graph $G$ according to $\ppath{p,n}$. We have:
    \begin{enumerate}
        \item $v_1=\sigma(p)=P,v_h=\sigma(n)$.
        \item $(\vec l_{v_1},\vec r_{v_1}) = (\lPC, +\omega)$.
        \item $(\vec l_{v_h},\vec r_{v_h}) \geq  (\vec l_n, +\omega)$.
        \item $(\vec l_{v_{i+1}},\vec r_{v_{i+1}}) = \mrm{Rel}_{(v_{i}, v_{i+1})} (\vec l_{v_{i}},\vec r_{v_{i}})$.
    \end{enumerate}

    Then, we shorten the sequence by removing unnecessary cycles. If there exists $i<j\leq h$ such that $v_i=v_j$ and $\vec l_{v_j} - \vec l_{v_i} \leq 0$, the subsequence $\{v_{i+1},\dots, v_j\}$ is useless and thus is removed. Let the new sequence be $\{v_1^\ast,v_2^\ast,...,v_{h^\ast}^\ast\}$. The new configuration $\vec l_{v_{h^\ast}^\ast}$ is no less than $\vec l_n$. Here are two cases:
    \begin{enumerate}
        \item $\{v_1^\ast,v_2^\ast,...,v_{h^\ast}^\ast\}$ is acyclic. We have $h^\ast\leq |\ChiSCC|$. Appending the fastest returning path (back to $P$, the starting symbol of $\CC$) after $\{v_1^\ast,v_2^\ast,...,v_{h^\ast}^\ast\}$ yields a pumping sequence as long as we restrict: \[b_1 \geq \lPC +\Delta\cdot |\ChiSCC| + 1.\]
        The length of the pumping sequence is bounded by $2\cdot |\ChiSCC|$. We require $b_2\geq 2\cdot |\ChiSCC|$.
        \item $\{v_1^\ast,v_2^\ast,...,v_{h^\ast}^\ast\}$ contains at least one cycle. Suppose the first simple cycle is $\{v_i^\ast,\dots,v_j^\ast\}$ with $i<j \leq h^\ast$. We have $i\leq |\ChiSCC|$ and $j\leq 2\cdot |\ChiSCC|$. We have $\vec l_{v_j^\ast} \geq \vec l_{v_i^\ast} +1$. Repeating the cycle $\{v_{i+1}^\ast,\dots,v_j^\ast\}$ after $\{v_1^\ast, v_2^\ast,\dots ,v_{i}^\ast\}$ for $(\lPC + \Delta\cdot |\ChiSCC| +1)$ times gives a sequence with length bounded by:
        \[b_3\defeq |\ChiSCC| + |\ChiSCC|\cdot(\lPC + \Delta\cdot |\ChiSCC| +1). \]
        Appending the fastest returning path back to $P$ yields a short pumping sequence. We require $b_2 \geq b_3 + |\ChiSCC|$.
    \end{enumerate}

    By the conversion from a pumping sequence to a pumping cycle, there exists a pumping cycle with length bounded by $b_2$. All the requirements above give reasonable $b_1,b_2\leq \exp(\size\CC)$.
\end{proof}

The two-sided pumpability case is more complicated. However, the equivalence between the existence of a pumping sequence and the existence of a pumping cycle still holds.

\LemmaTwoSidedPumping*
\begin{proof}
    We distinguish the exponential bounds. We refer to $\vec l_{n_i}\geq b_1, \vec r_{n_j}\geq b_2$ and $|\ppath{x,y}|\leq b_3$.
    
    Let $h = |\ppath{p,q}|$. Nodes on $\ppath{p,q}$ are denoted by the sequence $\{n_1,n_2,\dots, n_h\}$ with $n_1 =p$ and $ n_h=q$. W.l.o.g., let $i$ be the smallest one with $\vec l_{n_i}\geq b_1$. Let $j$ be the smallest one with $\vec r_{n_j}\geq b_2$. Assume that $i\leq j$. By this assumption, for every $m\leq i$, we can infer that $(\vec l_{n_m},\vec r_{n_m}) \leq (b_1,b_2)$. By the most-simplified assumption of derivation trees, we have:
    \[i\leq |\ChiSCC|\cdot(b_1+1)\cdot (b_2+1). \]

    Then, we construct the pumping sequence $\{v_1,v_2,\dots\}$ on the forward-coverability relation extension graph $G$. Let $v_1 = \sigma(n_1),\dots,v_{j} = \sigma(n_{j})$. We have:
    \begin{enumerate}
        \item $(\vec l_{v_1},\vec r_{v_1}) = (\lPC,\rPC)$.
        \item For $m\in [1,j-1]$, it holds: \[(\vec l_{v_{m+1}},\vec r_{v_{m+1}}) = \mrm{Rel}_{(v_{m}, v_{m+1})} (\vec l_{v_{m}},\vec r_{v_{m}}).\]
    \end{enumerate}

    Next, we construct a new sequence $\{v_1^\ast,\dots,v_{h^\ast}^\ast\}$. Initialize the new sequence by $h^* = i-1$ and $v_1^\ast=v_1,\dots,v_{i-1}^\ast = v_{i-1}$. Scan $v_m$ for $m\geq i$ iteratively:
    \begin{enumerate}
        \item Add $v_m$ to the bottom. Now $v^\ast_{h^\ast} = v_m$.
        \item If there exists $m^\ast\in[i,h^\ast-1]$ such that $v^\ast_{m^\ast} = v^\ast_{h^\ast}$, the subsequence $\{v^\ast_{m^\ast},\dots, v^\ast_{h^\ast}\}$ forms a cycle. The right effect of this cycle is $\gamma = \vec r_{v^\ast_{h^\ast}} - \vec r_{v^\ast_{m^\ast}}$. If $\gamma \leq0$, the cycle is useless. We remove $\{v^\ast_{m^\ast+1},\dots, v^\ast_{h^\ast}\}$.
        \item Otherwise, $\gamma >0$. We found a simple cycle that lifts the right configuration by at least $1$. It holds that $h^\ast - i +1\leq |\ChiSCC|$. We define the following bound:
        \[b_4 \defeq \rPC + \Delta\cdot |\ChiSCC| + 1.\]
        Repeat the cycle $\{v^\ast_{m^\ast+1},\dots, v^\ast_{h^\ast}\}$ for $b_4$ times after the sequence $\{v^\ast_{1},\dots, v^\ast_{m^\ast}\}$. The right configuration is no less than $b_4$. We restrict $b_1$ by:
        \[b_1 \geq \lPC + \Delta \cdot|\ChiSCC|\cdot (b_4 + 2)+1.
        \]
        Appending the fastest returning path back to $P$ yields a pumping sequence. It has a length bound $i + |\ChiSCC|\cdot(b_4+2)$. We require $b_3 \geq i + |\ChiSCC|\cdot(b_4+2)$.
    \end{enumerate}

    If we can not find any cycle with $\gamma>0$, all cycles after $v^\ast_{i-1}$ are removed. The sequence $\{v^\ast_{1},\dots, v^\ast_{h^\ast}\}$ is short. We have $h^\ast < i + |\ChiSCC|$. Since we removed cycles with $\gamma \leq 0$, we have $\vec r_{v^\ast_{h^\ast}} \geq b_2$ by the monotonicity of $\delta_A(\cdot),\delta_B(\cdot)$. In this case,  we require:
    \begin{equation*}\begin{aligned}
        b_1 &\geq \lPC + 2\cdot \Delta\cdot |\ChiSCC|+1,\\
        b_2 &\geq \rPC + \Delta \cdot |\ChiSCC| +1.
    \end{aligned}\end{equation*}
    Appending the fastest returning path after $\{v^\ast_{1},\dots, v^\ast_{h^\ast}\}$ yields a pumping sequence with length bounded by $i + 2\cdot |\ChiSCC|$. We require $b_3 \geq i + 2\cdot |\ChiSCC|$.

    By the conversion of a pumping sequence to a pumping cycle, there exists a pumping cycle with length bounded by $b_3$. None of the requirements above conflict with another. We have $b_1,b_2,b_3\leq \exp(\size\CC)$.
\end{proof}

Next, we prove the decidability of the $\exp$-forward-pumpability.

\LemmaDecidabilityofPumping*

\begin{proof}
    Consider the two-sided forward-pumpability. We have $\lPC,\rPC\in \Pump$.
    If there is a pumping cycle $\cyc{x,y}$ with its length $h\defeq |\ppath{x,y}|$ bounded by $\exp(\size\CC)$, then there exists a pumping sequence $\{v_1,\dots,v_{h}\}$ on the forward-coverability relation extension graph $G$. If there exists any $v_i$ with $\vec l_{v_i} \geq \lPC+h\cdot \Delta + 1$, we know that $\CC$ is $\exp$-forward-pumpable on $\lPC$ by Lemma \ref{Lemma:Lower bound for delta}. The case for $\vec r_{v_i}$ and $\rPC$ is symmetric.
    
    Instead of constructing $G$, we construct the truncated forward-coverability relation extension graph $G'$ by replacing function $\delta_A(s)$ (and $\delta_B(s)$ for the right-side case) in $G$ with the truncated one: \[\max\left\{\delta_A(s), \lPC+h\cdot \Delta + 1\right\}.\]
    
    The truncated graph $G'$ is effectively computable, since the coverability problem of 1-GVAS is in \EXPSPACE \ \cite{DBLP:conf/icalp/LerouxST15_Coverability}. Then we enumerate cycles on $G'$ with length bounded by $h$ to check the pumpability. The overall complexity is within an exponential space of $\size\CC$. Since $\size\CC$ is an estimation of the size under unary encoding, we conclude that the $\exp$-pumpability is decidable in \EXPSPACE{} under unary encoding.
\end{proof}

As for the combinatorial decomposition, we first prove the size bound of the new division set. The conclusion is slightly different (the bound becomes $\exp(\size\CC)$, instead of $\exp(\size\KT)$), since we require $\KT$ be fully constrained. The values of $\lPC,\rPC$ are taken into account in $\size\CC$.

\LemmaSizeofDivisionTwo*
\begin{proof}
    In the configuration encoding step, the set of nonterminals is bounded by $|\ChiSCC(\CC)| \cdot (v_{\mrm{B}}+1)$ where $v_{\mrm{B}}=\exp(\size\CC)$. Therefore, the SCC division step produces at most $\exp(\size\CC)$ strongly-connected $\seg{p^\dagger,q^\dagger}$. In the subtree Completion step, each of these strongly-connected $\seg{p^\dagger,q^\dagger}$ has at most one uncovered child. By Lemma \ref{lemma:small_division}, this uncovered child produces at most $\exp(\size\CC)$ segments. The overall segments are no more than $\exp(\size\CC)$.
\end{proof}

For Lemma~\ref{lemma:combinatorial decomposition}, it remains to show the strict decrease in rank.

\LemmaCombinatorialDecomposition*
\begin{proof}
    In the main text, we obtained the set $\subdiv{\seg{p,q}}$. For $\seg{p^\ddagger,q^\ddagger
    }\in \subdiv{\seg{p,q}}$, we constructed $\CC^\ddagger \capture \seg{p^\ddagger, q^\ddagger}$. Provided that $\ppath{p^\ddagger,q^\ddagger} \subseteq \ppath{p,q}$, we have to prove $\rank{\CC^\ddagger} < \rank{\CC}$. We claim that $\gdim{\CC^\ddagger} < \gdim{\CC}$.

    If $\gdim{\CC} =2$, then the claim is immediate, since $\CC^\ddagger$ is left-orthogonalized. Its geometric dimension is at most $1$.

    If $\gdim{\CC} =1$, we will prove the triviality of $\CC^\ddagger$ by contradiction. Assume that $\CC^\ddagger$ is not trivial. Consequently, the set of possible top cycles $ {\mathcal{O}}_{\CC^\ddagger}$ is not empty. For any top cycle $\cyc{x^\ddagger,y^\ddagger}$ in $ {\mathcal{O}}_{\CC^\ddagger}$, it holds $\sigma(x^\ddagger) =\sigma(y^\ddagger)$. We denote $\sigma(x^\ddagger) = (_{\vec l_X}X)$. This top cycle can be mapped to a top cycle $\cyc{x,y}$ in $ {\mathcal{O}}_{\CC}$. It holds that $\vec l_x = \vec l_X$ and $\vec l_y = \vec l_X$. Equivalently, we have $\Delta(\cyc{x,y})[1] = 0$.
    
    Since the dimension of $\bbq(\Delta( {\mathcal{O}}_{\CC}))$ is $1$, this vector space can be represented as $\bbq( \{\lambda\})$ for some $\lambda \in \bbq^2, \lambda \neq (0,0)$. Since the original KLM component $\CC$ is not left-orthogonal, we have $\lambda[1]\neq 0$. Since $\Delta(\cyc{x,y}) \in \bbq(\Delta( {\mathcal{O}}_{\CC}))$, we can infer that $\Delta(\cyc{x,y})[2] = 0$. This indicates that all top cycles in $ {\mathcal{O}}_{\CC}$ have effect $(0,0)$, contradicting the nontrivial assumption.
\end{proof}

\section{Proofs in Section \ref{section:perfect_KLM_tree}}
\label{section:appendix_perfect_KLM_tree}
In this section, we prove Lemma~\ref{lemma:KT as certificate}. We show that a perfect KLM tree itself is a certificate to the reachability problem. That is, given a KLM tree $\KT$ without knowing the derivation tree $\TT$ it captures, we are still able to construct another complete, nonnegative derivation tree $\NTT$. To accomplish this, we introduce the following concepts and techniques first.

\emph{Appending segments.} Given two segments $\seg{x, y}, \seg{x^\dagger, y^\dagger}$. If $\sigma(y) = \sigma(x^\dagger)$, then the segment $\seg{x^\dagger, y^\dagger}$ can be appended after $\seg{x, y}$. This operation replaces the leaf node $y$ by $\seg{x^\dagger, y^\dagger}$. Given a cycle $\cyc{x,y}$, we denote by $\cyc{x,y}^m$ the cycle obtained by appending $\cyc{x,y}$ to itself for $(m-1)$ times. We call $\cyc{x,y}^m$ an \emph{$m$-duplication} of $\cyc{x,y}$.

\emph{Non-negativity of duplications.} Consider the $m$-duplication of a cycle: $\cyc{x,y}^m$. Let the $m$ identical cycles be: \[\cyc{x_1, y_1}, \dots ,\cyc{x_m,y_m}\] where $y_{i-1} = x_i$. We fix the root configuration $(\vec l_{x_1}, \vec r_{x_1}) \geq (0,0)$ and the leaf configuration $(\vec l_{y_m}, \vec r_{y_m}) \geq (0,0)$. We claim that $\cyc{x,y}^m$ is a nonnegative cycle, if the following conditions hold: 
\begin{enumerate}
    \item For any node $n_1\in \cyc{x_1,y_1}$, we have $(\vec l_{n_1},\vec r_{n_1}) \geq (0,0)$.
    \item For any node $n_m\in \cyc{x_m,y_m}$, we have $(\vec l_{n_m},\vec r_{n_m}) \geq (0,0)$.
\end{enumerate}

To see this, consider any node $n_i \in \cyc{x_i,y_i}$ for $i\in [1,m]$. Its configuration $(\vec l_{n_i},\vec r_{n_i})$ is the linear i
nterpolation of $(\vec l_{n_1},\vec r_{n_1})$ and $(\vec l_{n_m},\vec r_{n_m})$. More precisely, we have:
\[(\vec l_{n_i},\vec r_{n_i}) = \frac{m-i}{m-1} (\vec l_{n_1},\vec r_{n_1}) + \frac{i-1}{m-1} (\vec l_{n_m},\vec r_{n_m}).\]
It is immediate that the configuration of $n_i$ is nonnegative.

\emph{Non-negativity of segments.} Given a segment $\seg{x,y}$. Denote the number of nodes as $|\seg{x,y}|$. Denote the largest absolute value as $\norm{\seg{x,y}}$. Intuitively, as long as the configuration of node $x$ (or $y$) satisfies:
\[\vec l_x,\vec r_x \geq |\seg{x,y}|\cdot \norm{\seg{x,y}},\]
then $\seg{x,y}$ is a nonnegative segment.

We make a review on the pumpability. For a nontrivial $\CC \in \KT$, we analysis the orthogonality of $\CC$ and the set $\Pump$, and focus on the $\exp$-pumping cycles.
\begin{enumerate}
        \item If $\CC$ is left-orthogonal, we check whether $\rPC \in \Pump$ and $\rQC \in \Pump$. If $\rPC\in \Pump$, by $\exp$-pumpability, there exists a forward-pumping cycle $\cyc{x_1,y_1}$ such that:
        \begin{itemize}
            \item $\cyc{x_1,y_1}$ has root configuration $(\vec l_{x_1}, \vec r_{x_1}) = (+\omega, \rPC)$.
            \item $\cyc{x_1,y_1}$ is nonnegative.
            \item $\cyc{x_1, y_1}$ lifts $\rPC$ by at least $1$. That is, the configuration of $y_1$ satisfies $(\vec l_{y_1},\vec r_{y_1}) \geq (+\omega, \rPC+1)$.
        \end{itemize}
        
        \item If $\CC$ is left-orthogonal and $\rQC \in \Pump$, there exists a backward-pumping cycle $\cyc{x_2,y_2}$ such that:

        \begin{itemize}
            \item $\cyc{x_2,y_2}$ has leaf configuration $(\vec l_{y_2}, \vec r_{y_2}) = (+\omega, \rQC)$.
            \item $\cyc{x_2,y_2}$ is nonnegative.
            \item $\cyc{x_2, y_2}$ lifts $\rQC$ by at least $1$. That is, the configuration of $x_2$ satisfies $(\vec l_{x_2},\vec r_{x_2}) \geq (+\omega, \rQC+1)$.
        \end{itemize}
        \item If $\CC$ is right-orthogonal, the case is symmetric.
        \item If $\CC$ is not orthogonal, it is possible that the pumping cycles are required to lift configurations on both sides. For example, suppose $\lPC, \rPC\in \Pump$, There exists a forward-pumping cycle $\cyc{x_1,y_1}$ such that:

        \begin{itemize}
            \item $\cyc{x_1,y_1}$ has root configuration $(\vec l_{x_1}, \vec r_{x_1}) = (\lPC, \rPC)$.
            \item $\cyc{x_1,y_1}$ is nonnegative.
            \item $\cyc{x_1, y_1}$ lifts both $\lPC, \rPC$ by at least $1$. It holds that $(\vec l_{y_1},\vec r_{y_1}) \geq (\lPC+1, \rPC+1)$.
        \end{itemize}
        If $\CC$ is not orthogonal and $\lQC,\rQC\in \Pump$, the case is also symmetric.
\end{enumerate}

Denote the forward-pumping cycle by $\theta_1 = \cyc{x_1,y_1}$ if such cycle exists. Otherwise, let $\theta_1 = \varnothing$. By this assumption, we denote the backward-pumping cycle by $\theta_2 = \cyc{x_2,y_2}$. To emphasize that the pumping cycles $\theta_1,\theta_2$ are from $\CC$, we write $\ccf, \ccb$.

\LemmaKTasCertificate*
\begin{proof}
    We construct such $\NTT$ from a solution of $\EE_{\KT}$. Let one of the minimal solution of $\EE_{\KT}$ be $\solmin$. Let the Hilbert basis (i.e., the set of all nonnegative, minimal, homogeneous solutions) be $H$. Recall the definition of the sum of solutions in $H$:
    \[\solh = \sum_{\sol\in H} \sol.\]
    
    Let $\#(\sol)$ be the vector of all Parikh image variables in solution $\sol$. That is, we omit all configuration variables in $\sol$. Moreover, when considering the Parikh image vector of $\CC$, we write $\#(\sol)(\CC)$. We denote the Parikh image vector of $\ccf,\ccb$ as $\#(\ccf),\#(\ccb)$. We have $\#(\solh)\geq \vec 1$, since all Parikh image variables are unbounded.
    
    Since $\KT$ is fully constrained, for any solution $\sol$ and any bounded configuration variable $v$, we have $\sol(v) = \solmin(v)$. We construct $\NTT$ according to the solution:
    \[ \solmin + b \cdot \solh. \]
    And the rest of the proof is to determine the coefficient $b$.

    \emph{Step 1: covering the pumping cycles.} For every nontrivial $\CC$, denote the number of nodes in $\ccf$ by $|\ccf|$ and the largest absolute value of terminals in $\ccf$ by $\norm{\ccf}$. We obtain the following bound $b_1$ by:
    \[b_1 = \max_{\CC \in \KT}\{|\ccf|\cdot \norm{\ccf}+ |\ccb| \cdot \norm{\ccb} \} + 1. \]
    Clearly, $b_1\cdot \#(\solh)$ covers the Parikh images of $\#(\ccf), \#(\ccb)$ for all $\CC$.

    \emph{Step 2: organizing the residual rules.} Consider the Parikh image $\#(\solmin)(\CC)$. Since it fulfills the Euler equations in $\EE_\CC$, it is the Parikh image of some segment. Denote the segment by $\seg{x_3,y_3}$ or $\gamma(\CC)$. We have $\sigma(x_3)= P(\CC)$ and $\sigma(y_3) = Q(\CC)$. Let the bound $b_2$ be: 
    \[b_2 = \max_{\CC\in \KT} \{|\gamma(\CC)|\cdot \norm{\gamma(\CC)}\} + 1, \]
    where $|\gamma(\CC)|,\norm{\gamma(\CC)}$ are defined analogously. Moreover, we denote: \[b_3 = b_1 \cdot (b_2+1).\]
    The residual Parikh image is defined as: \[\#_{\mrm{res}}(\CC) \defeq b_3 \cdot \#(\solh)(\CC) - \#(\ccf) - \#(\ccb),\]
    which is the Parikh image of a top cycle in $ {\mathcal{O}}_\CC$. To see this, notice that every production rule in $\RhoSCC,\RhoL,\RhoR$ appears at least once in $\#_{\mrm{res}}(\CC)$ since $b_3 > b_1$. Moreover, for any nonterminal in $\#_{\mrm{res}}(\CC)$, its in-degree and out-degree coincide. Denote the top cycle with the Parikh image $\#_{\mrm{res}}(\CC)$ by $\cyc{x_4, y_4}$ or $\ccr$. Since $\sigma(x_4)$ and $\sigma(y_4)$ can be any nonterminal in $\ChiSCC$, we assume that $\sigma(x_4) = \sigma(y_4) = Q(\CC)$. The upper bound of the sizes of $\ccr$ for all $\CC$ is given by:
    \[b_4 = \max_{\CC\in\KT} \{|\ccr| \cdot \norm{\ccr}\}.\]

    \emph{Step 3: the pumping step.}  We construct the strongly-connected segment $\seg{p,q}$ with $\CC \capture \seg{p,q}$ by appending segments as:
    We are going to decide how many times $\ccf,\ccb,\ccr$ are duplicated in order to guarantee the non-negativity of $\seg{p,q}$. Let the bound $b_5$ be:
    \[b_5 = (b_2+b_4).\]
    We claim that:\[\seg{p,q} = \ccf^{b_5} \gamma(\CC) \ccr^{b_5} \ccb^{b_5}\] is a nonnegative segment. Consequently, we have: \[b = b_3 \cdot b_5\] according to the Parikh image.
    
    We check the configuration non-negativity with case analysis as follow:
    
    \emph{On the orthogonal side.} Consider the left-orthgonal case. All left configurations in $\ppath{p,q}$ are hard-encoded. For a left-degenerate rule $(_{\vec l_X}X)\rightarrow A(_{\vec l_Y}Y) \in \RhoSCC(\CC)$, we have the certificate rule $A\rightarrow \vec l_Y-\vec l_X \in \RhoL(\CC)$. We replace certificate rules in $\seg{p,q}$ with a complete, nonnegative derivation tree. All configurations to the left of $\ppath{p,q}$ are nonnegative. 
    
    \emph{On the non-orthogonal side.} We will illustrate the non-negativity for $\CC$ which is neither left- nor right-orthogonal. We discuss the boundedness of $\lPC$ and $\rPC$. If $\lPC$ is unbounded, we have $\lPC \notin \Pump$ and:
    \begin{equation*}\begin{aligned}
        \vec l_p &= \solmin(\lPC) + b\cdot \solh(\lPC) \\
        &\geq \solmin(\lPC) + b.
    \end{aligned}\end{equation*}
    Otherwise, we have $\lPC \in \Pump$ and $\vec l_p = \solmin(\lPC)$. The case is symmetric for $\rPC$.
    
    First, we prove the configuration non-negativity of $\ccf^{b_5}$. W.l.o.g., suppose that $\lPC$ is unbounded and $\lPC$ is bounded. That is, $\rPC \in \Pump$ and $\lPC\notin \Pump$. The configuration satisfies: \[(\vec l_p, \vec r_p) \geq (\solmin(\lPC) + b, \solmin(\rPC)).\]
    At the end of $\ccf^{b_5}$, the configuration is at least: \[(\solmin(\lPC) + b - b_1\cdot b_5, \solmin(\rPC) + b_5).\]
    The lower bound above is no less than $(b_1 \cdot b_2\cdot b_5, b_5)$. 
    
    Since $\ccf$ is a pumping cycle, it is always repeatable if the starting configuration is $(+\omega, \rPC)$. In fact, any duplication of $\ccf$ guarantees the non-negativity of configurations to the right of $\ppath{p,q}$. Moreover, by the following inequality:
    \[|\ccf|\cdot \norm{\ccf} \cdot b_5 \leq b_1\cdot b_5 \leq b_1\cdot b_2\cdot b_5,\] 
    the $b_5$-duplication of $\ccf$ is also nonnegative for any configuration to the left of $\ppath{p,q}$. We conclude that $\ccf^{b_5}$ is nonnegative.

    Because of the same reason, the duplication $\ccb^{b_5}$ is also nonnegative. In general, regardless of whether the variables $\lPC, \rPC,\lQC,\rQC$ belong to $\Pump$, the configuration at the end of $\ccf$, or symmetrically, at the beginning of $\ccb$, is at least $(b_5,b_5)$.

    Next, we show the non-negativity in $\gamma(\CC)$.The configuration at the beginning of $\gamma(\CC)$ is at least $(b_5, b_5)$. Since: \[|\gamma(\CC)| \cdot \norm{\gamma(\CC)}\leq b_2 \leq b_5,\] the segment $\gamma(\CC)$ is also nonnegative. The configuration at the end of $\gamma(\CC)$ is at least:
    \[(b_5- b_2, b_5-b_2 ) = (b_4, b_4).\]

    At last, we prove the non-negativity of the duplication $\ccr^{b_5}$. Denote these identical cycles by the sequence:
    \[\{(\ccr)_1, \dots ,(\ccr)_{b_5}\}.\]
    It suffices to show the non-negativity of $(\ccr)_1$ and $(\ccr)_{b_5}$. At the beginning of $(\ccr)_1$, the configuration is at least $(b_4,b_4)$. At the end of $(\ccr)_{b_5}$, the configuration is at least $(b_5,b_5)$. Immediately, it follows from:
    \[ |\ccr|\cdot \norm{\ccr} \leq b_4 \leq b_5\]
    that both $(\ccr)_1$ and $(\ccr)_{b_5}$ are nonnegative.

    We have constructed $\seg{p,q}$ for every nontrivial $\CC$ such that $\CC \capture \seg{p,q}$. $\seg{p,q}$ is nonnegative, and its strong connectivity directly follows from the construction: the top cycle $\ccr$ covers all nonterminals in $\ChiSCC$. 

    The rest of the KLM tree is easy to obtain. We construct a trivial segment for each trivial KLM component. We connect KLM components and terminals by the exit rules ($\rhoD$). Let $\NTT$ be the complete derivation tree under the solution:
    \[\solmin + b\cdot \solh.\]
    We conclude that $\NTT$ is nonnegative and $\KT \capture \NTT$.
\end{proof}

\end{document}